\documentclass[runningheads]{llncs}
\usepackage[T1]{fontenc}
\usepackage{graphicx}
\usepackage{amsmath}
\usepackage{amsfonts}
\usepackage{amssymb}
\usepackage{mathtools}
\usepackage{tikz}
\usepackage{hyphenat}
\usepackage{nicefrac}
\usepackage{booktabs}
\usepackage{dsfont}
\usepackage{todonotes}
\usepackage{bm}
\usepackage{BOONDOX-calo} 
\usepackage{authblk}
\usepackage{rotating}
\usepackage{multirow}
\usepackage{subcaption}
\usepackage{enumitem}

\usetikzlibrary{
    automata, 
    positioning, 
    calc, 
    shapes, 
    arrows, 
    fit
    }

\definecolor{PastelGreen}{HTML}{6ECB63}
\definecolor{FernGreen}{HTML}{116530}
\definecolor{CadmiumGreen}{HTML}{046C41}
\definecolor{Burgundy}{RGB}{141, 2, 31}
\definecolor{Salmon}{RGB}{250, 128, 114}

\tikzset{
	a/.style={
        draw,
        rounded corners=2pt, 
        minimum width=1cm, 
        minimum height=1cm, 
        thick,
        color=black,
        fill=CadmiumGreen, 
        text=white,
        font=\Large
        },
	b/.style={
        draw,
        rounded corners=2pt, 
        minimum width=1cm, 
        minimum height=1cm, 
        thick,
        color=black,
        fill=Burgundy, 
        text=white,
        font=\Large
        }
    }

\renewcommand{\phi}{\varphi}

\renewcommand{\epsilon}{\varepsilon}
\renewcommand{\theta}{\vartheta}

\newcommand{\vote}{\mathit{Vote}}

\begin{document}
\title{
    You All Need to Hear This: Wisdom of Crowds with Deliberating Jurors
    }
%
%
\author{%
    Rosa Bauer\inst{1}
    \and
    Adrian Haret\inst{2}
    }
\authorrunning{
    Bauer \& Haret
    }
%
\institute{%
    LMU Munich 
    \and 
    Royal Holloway, University of London\\
    \email{rosa.bauer@icloud.com}, 
    \email{adrian.haret@rhul.ac.uk}
    }
\maketitle              
\begin{abstract}
    We analyze a formal model of deliberation in which 
    agents have probabilistic access to evidence
    and selectively disclose it to their peers in order to induce opinion change. 
    We present theoretical and simulation results that tackle the question: 
    does deliberation improve the collective accuracy of the group? 
    Through a jury theorem we chart the conditions that allow wisdom of crowds
    to emerge after a few steps of deliberation, and, through simulations, 
    we show how this effect is sensitive to the evidence acquisition mechanism
    and, implicitly, to the diversity in the initial evidence distribution.

\end{abstract}

\section{Introduction}\label{section:introduction}
A persistent worry about democratic decision-making is that it relies on flawed inputs:
elections aggregate whatever opinions voters happen to express, without 
any guarantee that such opinions are well-informed or unbiased. 
The result is that the collective choice is often perceived to be sub-optimal, 
and trust in democratic institutions is eroded \cite{Brennan2016t,Caplan2008w}.

Deliberation is often proposed as a remedy for this problem 
\cite{Bachtiger2024s,Dryzek2019q,Estlund2018r,Estlund2008a,Fishkin2025f,Hong2004u,Landemore2013z}. 
The positive case for deliberation is that 
structured communication and weighing of reasons exposes
decision makers to competing arguments and new evidence,
thereby improving the quality of their opinions and, 
down the line, the quality of the group decision.
Deliberation, in this view, is a catalyst for the \emph{wisdom of the crowds}: 
it makes it possible for relevant knowledge and evidence, 
which is otherwise scattered across people and institutions, 
to be pooled together and made available to everyone. 
In this way, when it comes time to vote, voters are better informed 
and more likely to make the right call. In other words, deliberation complements 
the aggregation of opinions with a collective learning process, 
or so the optimistic story goes.

Critics of deliberation point out that such positive effects occur 
only in idealized conditions, rarely met in practice, 
and that whatever benefit deliberation may bring, its impact is mild at best.
On the ground, discussion has been observed to amplify biases, 
increase polarization and steer groups 
\emph{away} from the truth rather than towards it 
\cite{Brennan2016t,Mendelberg2002p,Stasser1985q,Sunstein2002m,Sunstein2014h}.

There is, thus, a tension at the heart of 
the deliberation project, with plausible arguments showing that 
deliberation can go either inspiringly right or spectacularly wrong. 
What is needed is a clearer picture of the conditions and incentives that are likely to 
tip the outcome of deliberation in one direction or the other.
Our aim is to provide such a picture.

In this paper we study a formal model of deliberation intended 
to capture the dynamics of information flow in deliberative scenarios.
In our model a set of agents face a decision over two alternatives ($a$ and $b$) 
and selectively disclose information to the group, 
motivated by the desire to correct the collective opinion if they disagree with it.
We assume a ground truth ranking of the alternatives, 
determined by the total amount of evidence supporting each. 
Agents form an opinion based on private evidence: this opinion is based on
the evidence they have access to at the time, 
such that an agent may support $b$ because they happen to have more evidence for it, 
even though $a$ is, objectively, the better alternative;
thus, agents may end up sharing evidence for an alternative they earnestly believe to be better,
even though they are wrong.
Deliberation is the process of agents sharing evidence with each other
and updating their opinions based on newly acquired information. 
This is a process that takes place over several rounds, with a natural end point 
when no one has anything new to share, or when there is unanimous agreement.
We are interested in the accuracy of the group post-deliberation, 
i.e., the probability that the collective decision at the end 
of deliberation is the superior alternative. In other words, we want to know: does deliberation make
the group wiser?

\paragraph{Contributions.}
We put forward a probabilistic process of evidence acquisition 
(the initial opinions agents start out with),
and trace how this evidence gets re-distributed via the public layer of communication
represented by the deliberative process.
Our results concern the truth-tracking properties of the ensuing dynamics.
We derive a jury-like theorem, 
i.e., a results describing the outcome of deliberation with high probability, 
as the number of agents grows to infinity. The deliberation process, 
in our basic setup, turns out to exhibit an interesting pattern: an initial majority 
for the better alternative, followed by a reversal due to the influx of evidence for other 
alternative; if, at this point, the optimal alternative has enough supporters left standing, 
they will rally and are likely to persuade the rest of the group 
to switch back to the superior alternative. 
We show that this is because, in the long run, all evidence is likely to become public.
We also find this conclusion to be sensitive to the initial distribution of evidence,
and computer simulations allow us to get a more granular 
understanding of the deliberation dynamics,
notably on how the salience of evidence affects the outcome of deliberation.

\paragraph{Related Literature.}
There is a vast literature on the various benefits of deliberation. 
These include an improved sense of legitimacy in democratic institutions
\cite{Caluwaerts2016r,Cohen1997h,Dryzek2001m,Germann2024h,Goodin2006r}, 
reduced polarization \cite{Caluwaerts2023l,Fishkin2025f,Gronlund2015r},
and the transformation of participants' preferences to a 
form that is less prone to aggregation paradoxes 
\cite{Farrar2010r,List2013m,Miller1992d}.
Here we are interested in the \emph{epistemic} benefits, 
specifically the idea that deliberation works as a 
decentralized procedure for improving collective decisions
\cite{Estlund2008a,Hong2004u,Landemore2013z,Landemore2014p}.
One factor thought to help a group's epistemic performance is \emph{diversity},
sometimes framed as more important than individual ability 
\cite{Hong2004u,Luan2012x,Page2008k,Reijula2021y}.
Other formal treatments of deliberation have been offered 
\cite{Dietrich2025k,Ding2021d,Goel2025d,Goyal2026l,List2011j,Siebe2024j},
with \cite{Chingoma2023i} being the most similar to ours. 
Our work can be seen as complementing the
latter model with a probabilistic component in the acquisition of evidence,
which allows us to derive results about the likelihood of correct deliberation outcomes.

Of course, there already exists a large body of work that argues 
for the wisdom of crowds via the Condorcet Jury Theorem (CJT)
\cite{Dietrich2023h,Goodin2018r,Grofman1983e,List2001g}.
The CJT, however, does not model the information flow that 
makes opinions change, and thus captures 
only a snaphot of the broader dynamic process we are interested in here.
Nonetheless, we see the CJT as a starting point and template for the type of results 
(jury-like theorems) we seek for deliberation.

\paragraph{Outline.}
Section \ref{section:model} lays out the basic deliberation model,
including the probabilistic evidence sampling mechanism,
followed by our jury theorem in Section \ref{section:jury-theorem}
and a variation of the model to accommodate evidence
of varying salience in Section \ref{section:biased-sampling}
Section \ref{section:conclusion} offers conclusions and thoughts for future work.
Full proofs are moved to the Appendix.
The code for the simulations is available at 
{\footnotesize\texttt{https://github.com/adrianharet/wisdom-probabilistic-deliberation-code}}.

\section{Model}\label{section:model}
We assume a set $N = \{1, \dots, n\}$ of $n$ \emph{agents} facing a binary decision problem
over two \emph{alternatives}, $a$ and $b$. There is a set $E$ of $m$ available 
\emph{evidence items} partitioned into $A = \{a_1, \dots, a_\ell\}$ and 
$B = \{b_1, \dots, b_{m-\ell}\}$, with $A \cup B = E$ and $A \cap B = \emptyset$.
The evidence items in $A$ and $B$ support alternatives $a$ and $b$, respectively,
and the ground truth ranking of the alternatives is determined
by the number of evidence items supporting them.
Without loss of generality we assume that $\ell > \nicefrac{m}{2}$, 
i.e., $a$ is taken to be (objectively) better than $b$ 
because it is supported by more evidence.
Alternatives $a$ and $b$ can be thought of as different options 
(e.g., bulding a highway or a park, voting for one candidate or another),
or the pro- and con- sides of a singular issue; in this framing evidence items 
are facts, arguments, or pieces of information in support of one of the alternatives.

Each agent $i$ has access to a subset of the total evidence, 
called $i$'s \emph{private evidence}.
The private evidence is indexed across a series of \emph{rounds} $t \in \{0, 1, \dots\}$,
and we write $E_i^t$ for the \emph{private evidence of $i$ at $t$}, 
with $E_i^t \subseteq E$.
Additionally, $A^t_i = E^t_i \cap A$ and $B^t_i = E^t_i \cap B$ are $i$'s sets of 
private evidence for $a$ and $b$, respectively, at $t$.
At each round $t$ agents use their private evidence to form an opinion about the alternatives. 
We write $v_i^t \in \{a, b, \emptyset\}$ for $i$'s \emph{opinion}, or \emph{vote}, 
at round $t$, with $v_i^t = x$, for $x \in \{a, b\}$, 
if $i$ has more private evidence (at $t$) for $x$ than for the other alternative, 
and $v_i^t = \emptyset$ if there is a tie in $i$'s evidence amounts.
If $v_i^t = \emptyset$ we say $i$ is \emph{indifferent} at $t$.  

The opinions are aggregated into a collective decision,
and we write $\vote^t(n)$ for the 
\emph{collective decision at round $t$ over a group of $n$ voters}, 
with $\vote^t(n) = x$, for $x \in \{a, b\}$, if there are more votes for $x$ than for 
the other alternative, and $\vote^t(n) = \emptyset$ in case of a tie.
In other words, $\vote^t(n)$ is the majority vote at $t$ among the non-indifferent voters
(indifferent voters are treated as abstentions).

\emph{Deliberation} consists of agents comparing their own opinion with the collective vote
and selectively sharing evidence, using the following bookkeeping rules.
First, we say that agent $i \in N$ \emph{dissents} at $t \geq 0$ from the collective opinion if 
$\vote^t(n) = a$ and $v_i^t = b$, or $\vote^t(n) = b$ and $v_i^t = a$, 
i.e., if $i$'s opinion is different from the collective decision and $i$ 
is not indifferent (if the collective vote is tied we assume, for convenience, that there are no dissenters).
Agents also keep track of disclosed evidence through a set $\mathit{PE}^t \subseteq E$ of 
\emph{public evidence at $t$}, which contains all items of 
evidence shared (by any agent) up to, and including, round $t$. 
Initial evidence, i.e., evidence sets at $t=0$ and hence initial opinions, 
are given by a probabilistic process (described below).

A round $t \geq 1$ of deliberation proceeds according to the idea that 
dissenters want to correct what they perceive to be a mistake in the collective opinion, 
so they share any undisclosed evidence that supports their own current opinion.
More precisely: 
agents enter round $t \geq 1$ knowing their opinion and the collective decision at $t-1$; 
every dissenting agent $i$ discloses all private evidence that is not already public 
(i.e., not in $\mathit{PE}^{t-1}$) and that supports their own opinion at $t-1$;
the set $\mathit{PE}^t$ is obtained by adding to $\mathit{PE}^{t-1}$ 
all the newly disclosed evidence; 
all agents update their private evidence with the newly disclosed evidence;
each agent $i$ then updates their vote from $v_i^{t-1}$ to $v_i^t$ 
based on the updated private evidence; 
finally, the collective decision $\vote^t(n)$ is recomputed based on the updated votes.
All updates are performed synchronously, i.e., agents are assumed to share and update 
their information at the same time.

If there is a round at which no evidence disclosure occurs, 
then agents' opinions (and with them the collective decision) stay the same forever after.
We say that deliberation \emph{terminates} at this point, and write 
$\vote^*(n)$ for the \emph{collective decision} at termination.
Deliberation can terminate for one of two reasons: 
either all agents agree with the collective opinion, 
or there are some remaining dissenters but none of them has any new evidence to disclose.
Note that in this dynamics public evidence can only ever increase, 
so there must be a point at which deliberation terminates.
Note, as well, that, in our setting the deliberation outcome depends entirely 
on the initial evidence distribution: given initial evidence sets $E_i^0$, for $i \in N$, 
the deliberation process is deterministic.

So how is evidence acquired initially? 
We assume this happens through a probabilistic model
that approximates a trial-and-error learning process (e.g., searching online)
repeated $k$ times (the amount of effort agents are willing to invest), 
with each search yielding a new random item of evidence.
Concretely, we fix a parameter $k \in \{1, \dots, m\}$ and have each agent 
$i \in N$ independently sample $k$ items from $E$, without replacement. 
Intuitively, we imagine a bag of $m$ evidence items,
of which $\ell$ (the `successes') support $a$ and the remaining ones 
(the `failures') support $b$,
with each agent acquiring their initial private evidence 
by independently reaching into the bag and drawing $k$ items without replacement.

Thus, the number $|A_i^0|$ of evidence items for $a$ that $i$
draws initially is described by the hypergeometric distribution 
\cite{Lehmann2022a}:
$$
    \left| A_i^0 \right| \sim \mathrm{Hypergeometric}\big(m, \ell, k \big).
$$
with $m$ corresponding to the population size, 
$\ell$ corresponding to the successes, 
and $k$ the number of draws in each trial. 
Thus, the probability that $i$ initially draws exactly $j \in \{0, \dots, k\}$
evidence items for $a$ can be thought of as the number of 
observed successes out of $k$ draws:
$$
    {\Pr}\Big[ \left| A_i^0 \right| = j \Big] = 
    \frac{\binom{\ell}{j}\binom{m-\ell}{k-j}}{\binom{m}{k}}.
$$

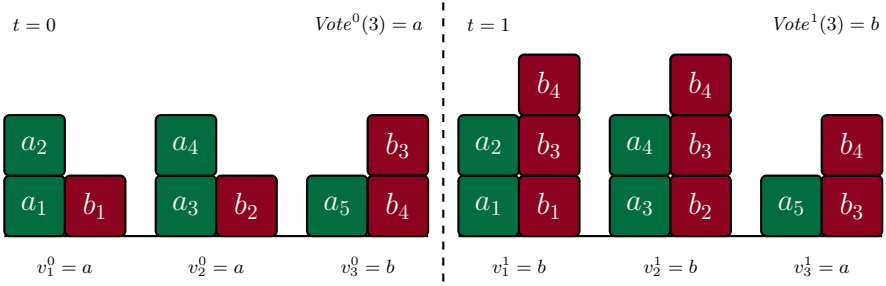
\begin{figure}[t]
    \centering
    \begin{tikzpicture}[scale=0.8, transform shape]
        \draw[thick] (0, 0) -- (7,0);
        \node at (0.5, 3.5) {$t=0$};
        \node at (6, 3.5) {$\vote^0(3)=a$};

        \node at (1, -0.5) {$v_1^0 = a$};
        \node[a] at (0.5, 0.5) {$a_1$};
        \node[a] at (0.5, 1.5) {$a_2$};
        \node[b] at (1.5, 0.5) {$b_1$};

        \node at (3.5, -0.5) {$v_2^0 = a$};
        \node[a] at (3, 0.5) {$a_3$};
        \node[a] at (3, 1.5) {$a_4$};
        \node[b] at (4, 0.5) {$b_2$};

        \node at (6, -0.5) {$v_3^0 = b$};
        \node[a] at (5.5, 0.5) {$a_5$};
        \node[b] at (6.5, 1.5) {$b_3$};
        \node[b] at (6.5, 0.5) {$b_4$};

        \draw[thick,dashed] (7.25, -0.75) -- (7.25, 4);

        \draw[thick] (7.5, 0) -- (14.5,0);
        \node at (8, 3.5) {$t=1$};
        \node at (13.55, 3.5) {$\vote^1(3)=b$};

        \node at (8.5, -0.5) {$v_1^1 = b$};
        \node[a] at (8, 0.5) {$a_1$};
        \node[a] at (8, 1.5) {$a_2$};
        \node[b] at (9, 0.5) {$b_1$};
        \node[b] at (9, 1.5) {$b_3$};
        \node[b] at (9, 2.5) {$b_4$};

        \node at (11, -0.5) {$v_2^1 = b$};
        \node[a] at (10.5, 0.5) {$a_3$};
        \node[a] at (10.5, 1.5) {$a_4$};
        \node[b] at (11.5, 0.5) {$b_2$};
        \node[b] at (11.5, 1.5) {$b_3$};
        \node[b] at (11.5, 2.5) {$b_4$};

        \node at (13.5, -0.5) {$v_3^1 = a$};
        \node[a] at (13, 0.5) {$a_5$};
        \node[b] at (14, 0.5) {$b_3$};
        \node[b] at (14, 1.5) {$b_4$};
    \end{tikzpicture}
    \caption{%
        A deliberation instance with three agents. 
        Evidence items are shown as labeled blocks, 
        with $a$-evidence in green and $b$-evidence in red.
        During the first round, agent $3$ persuades agents 
        $1$ and $2$ to change their mind.
        There is more evidence for $a$ overall, but after one 
        round of deliberation the group goes from voting for $a$ 
        to unanimously supporting the inferior alternative $b$.
        }
        \label{fig:illustrative-deliberation-example}
\end{figure}

\begin{example}\label{ex:deliberation-example}
    Consider evidence sets $A = \{a_1, \dots, a_5\}$ and $B = \{b_1, \dots, b_4\}$,
    supporting alternatives $a$ and $b$, respectively,
    and three agents who draw three items of evidence each ($k=3$),
    and whose initial private evidence (at $t=0$) is depicted in 
    Figure \ref{fig:illustrative-deliberation-example}.
    The collective decision at $t=0$ is $\vote^0(3) = a$, 
    based on two votes for $a$ (agents $1$ and $2$) 
    and one vote for $b$ (agent $3$): the majority decision,
    at this point, is in favor of the superior alternative $a$.
    During the first round of deliberation, 
    agent $3$, who disagrees with the collective opinion
    and wants to correct a perceived error, 
    discloses their private evidence for $b$, i.e., evidence items $b_3$ and $b_4$.
    Agents $1$ and $2$ update their evidence and proceed to recalculate their beliefs.
    There is now unanimous agreement in favor of $b$! 
    Once such universal agreement is achieved
    we see no more disclosure (no one disagrees with the collective opinion,
    hence there is nothing to correct), and deliberation terminates with $b$ as the consensus decision. 
\end{example}

\noindent
Example \ref{ex:deliberation-example} shows that deliberation can take a correct majority 
and, worryingly, turn it into a unanimous 
decision for the inferior option: a motivated minority with a lot of evidence 
for their preferred alternative can persuade the rest of the group to their side;
and if this minority happens to be in favor of a bad choice, deliberation ends 
up amplifying the error, rather than correcting it.
What is more, previous work, both theoretical and experimental, 
has shown this can happen for a range of initial 
evidence distributions, as well as for alternative ways of organizing 
deliberation \cite{Chingoma2023i,Stasser1985q}.
It seems that there is ample space for deliberation to go wrong.
So is deliberation essentially harmful to the wisdom of crowds?
In the following sections we look at this question
by trying to understand not so much \emph{whether} things can go wrong (we know they can), 
but \emph{how likely} it is that they do.

\section{A Jury Theorem in Three Acts}\label{section:jury-theorem}
Our main result in this section is a jury-like theorem for the success of deliberation, 
obtained by collecting snapshots of the collective decision across 
different rounds of the deliberative process.
We are particularly interested in where deliberation ends up, 
i.e., in the probability ${\Pr}{\left[ \vote^*(n) = a \right]}$
of a correct majority vote at termination. 
The high-level argument of this section is that, as $n$ goes to infinity, 
the deliberation dynamics introduced in Section \ref{sec:model} 
plays out over at most three rounds.

\paragraph{The initial round.}
At $t=0$ each agent $i \in N$ forms an initial opinion 
based on the $k$ evidence items 
sampled independently from the bag of total evidence.
The probability of $i$ forming a correct opinion at this stage 
is the probability that $i$ samples (strictly) more evidence for $a$ than for $b$,
equivalent to $i$ sampling more than $\nicefrac{k}{2}$ items of evidence for $a$:
$$
    {\Pr}\Big[ v_i^0 = a\Big] = 
    {\Pr}{\left[ |A_i^0| > \frac{k}{2} \right]} = 
    \sum^k_{j = \lfloor \nicefrac{k}{2}\rfloor + 1}{\Pr}\Big[ |A_i^0| = j \Big].
$$
Also known as the \emph{survival function} of the hypergeometric distribution
(with threshold $\nicefrac{k}{2}$), 
we call this probability $i$'s \emph{initial accuracy}.
Since all agents draw evidence independently according to the same probabilistic model, 
everyone has the same initial accuracy, a quantity we denote by $p_a$.
Similarly, we write $p_b$ and $p_{\emptyset}$ for the probabilities
that an agent ends up voting for $b$ ($|A_i^0| < \nicefrac{k}{2}$) 
and, respectively, being indifferent ($|A_i^0| = \nicefrac{k}{2}$) at $t=0$.
These are the only possibilities, so $p_a + p_b + p_{\emptyset} = 1$.

The first step of our argument is to understand how 
$p_a$, $p_b$ and $p_\emptyset$ 
fluctuate depending on the parameters of the model.
Towards this, a useful concept is that of 
a \emph{median} $M$ of a random variable $X$, defined as a value such that
${\Pr}{\big[ X \leq M \big]} \geq \nicefrac{1}{2}$ and
${\Pr}{\big[ X \geq M \big]} \geq \nicefrac{1}{2}$.
For the hypergeometric distribution, medians can be bounded using 
the mean $\mu$ \cite{cooper_linearly_2010,siegel_median_2001}.

\begin{lemma}\label{lemma:hypergeometric-median-inequality}
    For a random variable $X$ with a hypergeometric distribution, 
    mean $\mu$ and median $M$, it holds that $\lfloor \mu \rfloor \leq M \leq \lceil \mu \rceil$. 
\end{lemma}

\noindent
Lemma \ref{lemma:hypergeometric-median-inequality} allows us to 
estimate initial accuracy, presented in the next result. 
The result focuses on $k$ as, conceptually, we think this is the main
variable in a deliberation context: $\ell$ and $m$ are considered to be fixed,
while $k$ approximates the amount of effort agents make prior to forming an
initial opinion.     

\begin{lemma}\label{lemma:round-0-agent-accuracies}
If $\nicefrac{\ell}{m} > \nicefrac{1}{2}$, it holds that: 
\begin{itemize}
    \item[\textnormal{($\textsf{a}$)}] if $k \geq 2(m-\ell)$, then $p_b = 0$;
    \item[\textnormal{($\textsf{b}$)}] if $k < 2(m-\ell)$, then:
    \begin{itemize}
        \item[\textnormal{($\textsf{b}_1$)}] if $k$ is odd, 
        then $0 = p_\emptyset < p_b < \nicefrac{1}{2} < p_a < 1$;
        \item[\textnormal{($\textsf{b}_2$)}] if $k$ is even, 
        $\nicefrac{\ell}{m} > \nicefrac{1}{2} + \nicefrac{1}{m}$
        and $k > \left(\nicefrac{\ell}{m} - \nicefrac{1}{2}\right)^{-1}$,
        then $p_{\emptyset} > 0$ and $0 < p_b < \nicefrac{1}{2} \leq p_a < 1$.
    \end{itemize} 
\end{itemize} 
\end{lemma}

\noindent
Intuitively, if $k \geq 2(m-\ell)$ there is simply not enough total evidence for $b$
to fill a $k$-item draw, and agents never vote for $b$.
If $k$ is odd there is no possibility of a tie in an agent's initial draw
(so agents can never be indifferent), and the overall evidence bias in favor of $a$
($\ell > \nicefrac{m}{2}$) is reflected in the draws.
If $k$ is even, the added possibility of indifferent agents 
pulls probability mass away from correct opinions, and we need 
stronger conditions for an initial vote for $a$: 
the bounds in Lemma \ref{lemma:round-0-agent-accuracies}$\textsf{b}_2$
push $\nicefrac{k}{2}$ below the mean $\lfloor \mu \rfloor$ of the hypergeometric
variables, and Lemma \ref{lemma:hypergeometric-median-inequality}
does the rest.

Thus, Lemma \ref{lemma:round-0-agent-accuracies} offers a (nearly) 
complete picture of initial accuracy as a function of $m$, $\ell$ and $k$. 
For an illustration see Figure \ref{figure:initial-accuracy};
note that accuracy there is: 
perfect for $k \geq 2(m-\ell) = 26$
(Lemma \ref{lemma:round-0-agent-accuracies}$\textsf{a}$);
above $\nicefrac{1}{2}$ for all odd $k$ 
(Lemma \ref{lemma:round-0-agent-accuracies}$\textsf{b}_1$);
drops for even $k$ (even below $\nicefrac{1}{2}$), 
and takes longer to get above $\nicefrac{1}{2}$, 
with the lower bound in 
Lemma \ref{lemma:round-0-agent-accuracies}$\textsf{b}_2$ shown as a vertical dashed line.

Individual accuracy is important mainly 
as a stepping stone on the way to collective accuracy, which 
we estimate next. Since ties in evidence amounts add case distinctions
but not significant insight, we tend to ignore them out by assuming $k$ is odd. 


\begin{proposition}\label{proposition:round-0-collective-outcome}
    If $\ell > \nicefrac{m}{2}$, it holds that:
    \begin{itemize}
        \item[\textnormal{($\textsf{a}$)}] if $k \geq 2(m-\ell)$, 
            then $\lim_{n \rightarrow \infty}{\Pr}\big[ \vote^0(n)  = a \big] = 1$
            and deliberation terminates at this round with $a$ as winner; 
        \item[\textnormal{($\textsf{b}$)}] if $k < 2(m-\ell)$ and $k$ is odd, then 
            $\lim_{n \rightarrow \infty}{\Pr}\big[ \vote^0(n)  = a \big] = 1$.
    \end{itemize}
\end{proposition}

\noindent
Intuitively, if $k \geq 2(m-\ell)$, then no one votes for $b$ 
(Lemma \ref{lemma:round-0-agent-accuracies}$\textsf{a}$);
so the collective decision is $a$; 
so no one dissents and deliberation terminates. 
For odd $k < 2(m-\ell)$ we get that 
$p_a > \nicefrac{1}{2}$ (Lemma \ref{lemma:round-0-agent-accuracies}$\textsf{b}_1$),
which means that, as $n \rightarrow \infty$, we can be certain of a majority for $a$.   
Proposition \ref{proposition:round-0-collective-outcome}$\textsf{b}$
is essentially the Condorcet Jury Theorem (CJT), showing up here 
as a special case of our model: the collective decision pre-deliberation, 
i.e., before people engage in any sort of communication.
This case can be thought of as a baseline for the success of deliberation.
An illustration of the result is shown in Figure \ref{figure:collective-accuracy} for $t=0$,
where the collective accuracy is virtually perfect for all odd $k$ below $2(m-\ell) = 80$.

\begin{figure}[t]
    \centering
    \begin{subfigure}[b]{0.49\textwidth}
        \includegraphics[width=\textwidth]{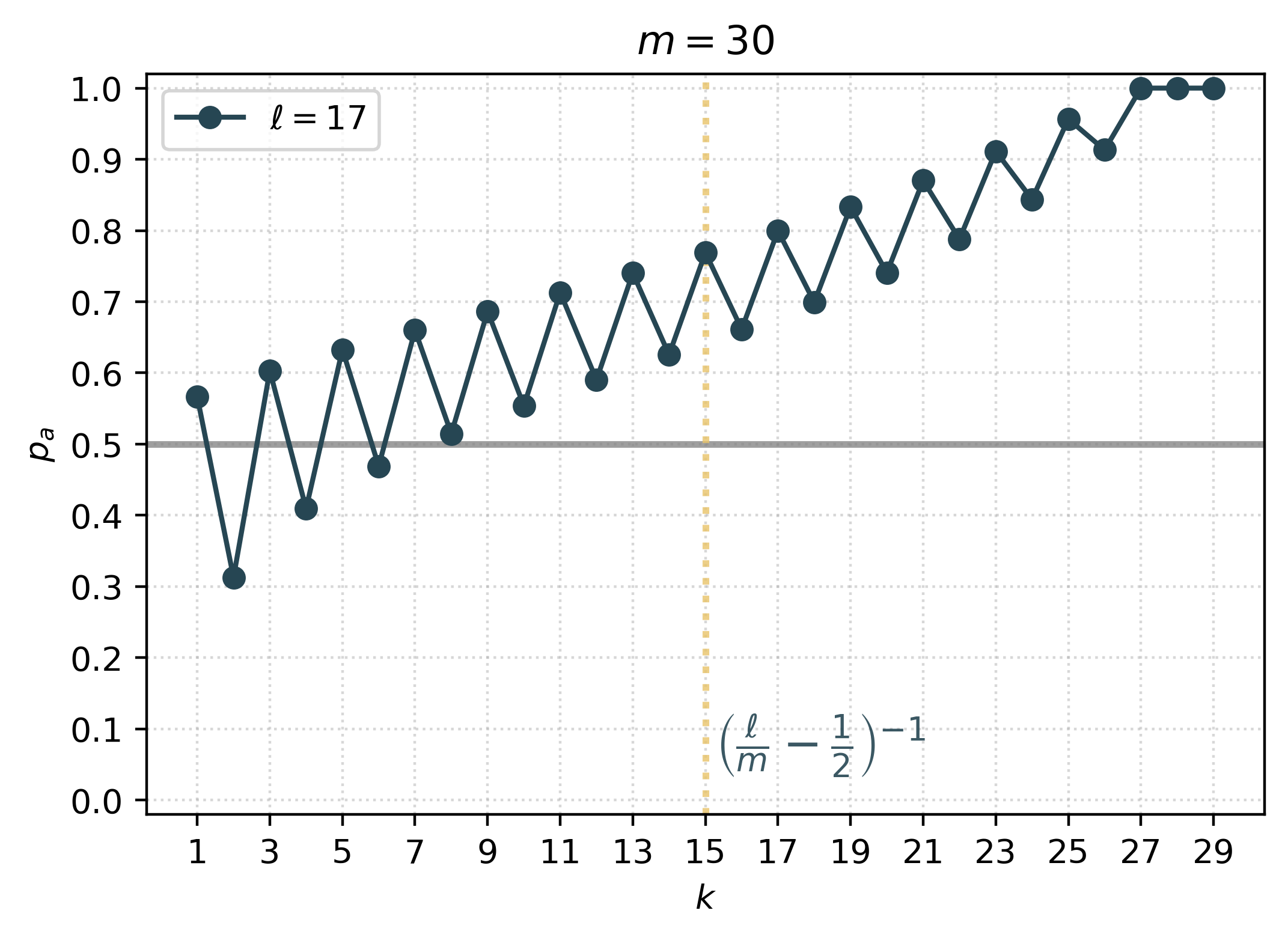}
        \caption{%
            Initial accuracy $p_a$ versus $k$ for $m=30$ and $\ell = 17$. 
            The crucial level $\nicefrac{1}{2}$ is shown as a horizontal line, 
            and the lower bound for even $k$ from
            Lemma \ref{lemma:round-0-agent-accuracies}($\textsf{b}_2$) 
            is shown as a vertical dashed line.  
            }
        \label{figure:initial-accuracy}
    \end{subfigure}
    \begin{subfigure}[b]{0.49\textwidth}
        \includegraphics[width=\textwidth]{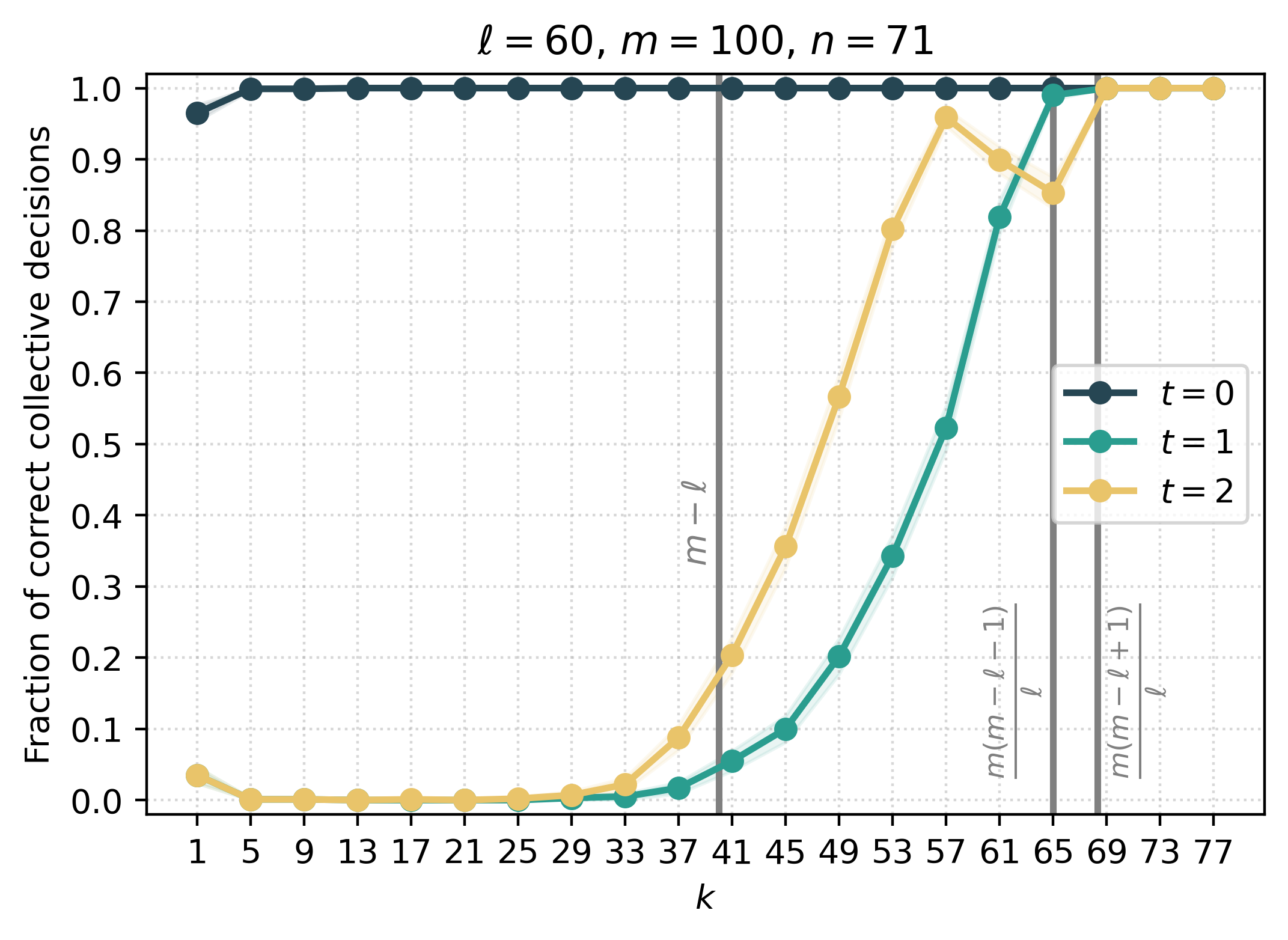}
        \caption{%
            Fraction of correct collective decisions 
            at rounds $t \in \{0, 1, 2\}$ versus $k$
            for $m=100$, $\ell = 60$ and $n=71$.
            The bounds obtained in 
            Propositions \ref{proposition:round-1-collective-outcome}
            and \ref{proposition:round-2-collective-outcome} 
            are shown as vertical dashed lines.
            }
        \label{figure:collective-accuracy}
    \end{subfigure}
    \caption{%
        Initial accuracy (on the left) as a plot of 
        the survival function of the hypergeometric distribution;
        and collective accuracy (on the right)
        as the fraction of correct decisions across 1000 simulations.
    }
    \label{figure:individual-and-collective-accuracy}
\end{figure}

\paragraph{Round one.}
Note that if $k \geq 2(m-\ell)$ there is no round one 
(Proposition \ref{proposition:round-0-collective-outcome}$\textsf{a}$),
so the following arguments work under the assumption that 
$k < 2(m-\ell)$ (also that $k$ is odd).

As with the previous round, our first aim is to estimate the probabilities 
of voting for $a$, $b$ and $\emptyset$ at the end of the round.
Note that, for an agent $i$, the probability of making a particular choice when round one ends
depends on $i$'s own, as well as \emph{everyone else's} initial evidence draw:
agent $i$ begins round one with their opinion from the previous round, but they might change it 
as other agents disclose some of their own private evidence.
Since we are interested in the limit behavior as $n$ grows to infinity, we
write $p'_a$, $p'_b$ and $p'_\emptyset$ for the probabilities
of $i$ voting for $a$, $b$ and $\emptyset$, respectively, at the end of round $1$, 
in the limit as $n$ grows to infinity:
$$
    p'_a = \lim_{n \rightarrow \infty}{\Pr}\big[ v^1_i = a \big], 
    \quad 
    p'_b = \lim_{n \rightarrow \infty}{\Pr}\big[ v^1_i = b \big], 
    \quad 
    p'_\emptyset = \lim_{n \rightarrow \infty}{\Pr}\big[ v^1_i = \emptyset \big].
$$
\noindent
Note that these probabilities still depend only on the initial distribution of evidence
and are the same for each agent (hence the ommission of an agent index).

Assuming, then, that we have made it to this point, 
let us take stock of where we are at the beginning of round one.
Under our assumptions, when round $t=1$ starts
there are ($i$) no indifferent agents, and
($ii$) (likely) \emph{some} votes for $b$ 
($p_b > 0$, by Lemma \ref{lemma:round-0-agent-accuracies}$\textsf{b}_1$),
with the exact number of $b$-voters being random but converging to $p_b n$
as $n \rightarrow \infty$ (the Law of Large Numbers).
Recall, as well, that $p_a > p_b$ and that a majority for $a$ is overhwelmingly likely at 
$t=0$ (Proposition \ref{proposition:round-0-collective-outcome}$\textsf{b}$).

Thus, round one very likely starts with $a$ as the reigning champion (decided at the previous round)
in a group consisting of a majority of $a$-supporters and a minority of dissenting $b$-supporters: 
as $n \rightarrow \infty$, we can effectively take for granted a collective vote for $a$ at $t=0$.
More formally, the probability that voter $i$ votes for alternative $x \in \{a, b, \emptyset\}$ 
at $t=1$ can be conditionalized on the collective vote at $t=0$:
\begin{align*}
    {\Pr}\Big[ v^1_i = x \Big] & = 
    {\Pr}\Big[ v^1_i = x \mid \vote^0(n) = a \Big] \cdot \underbrace{{\Pr}\Big[ \vote^0(n) = a \Big]}_{\rightarrow 1, \text{ as } n \rightarrow\infty} + \\
    & \qquad\qquad {\Pr}\Big[ v^1_i = x \mid \vote^0(n) = b \Big] \cdot \underbrace{{\Pr}\Big[ \vote^0(n) = b \Big]}_{\rightarrow 0, \text{ as } n \rightarrow\infty}.
\end{align*}
By Proposition \ref{proposition:round-0-collective-outcome}$\textsf{b}$, 
the probability of an initial majority for $a$ goes to $1$ as $n$ grows to infinity, 
so the second term above (the probability of an initial majority for $b$) becomes $0$.
Thus, the probability that $i$ votes for $x$ at the end of the first round is dominated by the first term,
and we can continue the argument under the assumption of an initial majority for $a$.

In this setup, the minority of $b$-supporters, thinking that a victory for $a$ is a grave error, 
aim to set the matter straight by disclosing all their evidence for $b$ during round one.
How much evidence for $b$ is this?
The key observation is that, as $n \rightarrow \infty$, 
this dissenting minority is very likely to hold, between them, 
\emph{all the existing evidence for $b$}. 
To state this formally, we write $N_a$ and $N_b$ for 
the sets of agents who end up voting for $a$ and $b$, respectively, at $t=0$.
Note that $N_a \cup N_b = N$ (no indifferent agents at $t=0$).

\begin{lemma}\label{lemma:initial-minority-covers-b-evidence}
    If $\ell > \nicefrac{m}{2}$ and $k$ is odd,
    $\lim_{n \rightarrow \infty}{\Pr}{\left[ \bigcup_{i \in N_b} B^0_i = B \right]} = 1$.
\end{lemma}

\noindent
Intuitively, this holds because, as $n$ grows
so does the number $|N_b|$ of $b$ voters 
(recall that $|N_b|$ is approximately $p_bn$, with $p_b > 0$);
thus, the probability that some item of evidence for $b$ is 
\emph{not} in any of the $B_i$'s for $i \in N_b$ decreases
and, in the limit, is 0.

The immediate effect of this disclosure event is that, once it happens,
every agent ends up knowing all $m-\ell$ evidence items for $b$.
From here, the question is whether enough $a$-supporters change their mind
to overturn the initial majority for $a$.
This, now, boils down to whether agents start out (at $t=0$)
with enough evidence for $a$ to withstand the shock of the disclosure of all $b$-evidence
at $t=1$:
$$
    p'_a = {\Pr}\Big[ |A_i^0| > m-\ell \Big], \quad
    p'_b = {\Pr}\Big[ |A_i^0| < m-\ell \Big], \quad
    p'_\emptyset = {\Pr}\Big[ |A_i^0| = m-\ell \Big].
$$
Finally, something we can calculate:
we are back to the survival function of the hypergeometric distribution, 
now with threshold $m-\ell$ instead of $\nicefrac{k}{2}$.

\begin{lemma}\label{lemma:round-1-agent-accuracies}
    If $\ell > \nicefrac{m}{2}$ and $k$ is odd, it holds that:
    \begin{itemize}
        \item[\textnormal{($\textsf{a}$)}] if $k \leq m-\ell$, then 
            $p'_a = 0$;
        
        \item[\textnormal{($\textsf{b}$)}] if $m-\ell < k < \nicefrac{m(m-\ell-1)}{\ell}$, 
            then
            $0 < p'_a < \nicefrac{1}{2} \leq p'_b < 1$
        \item[\textnormal{($\textsf{c}$)}] if $\nicefrac{m(m-\ell+1)}{\ell} < k < 2(m-\ell)$, 
            then
            $0 < p'_b < \nicefrac{1}{2} \leq p'_a < 1$.
    \end{itemize}
\end{lemma}

\noindent
Intuitively, if $k \leq m-\ell$, then at $t=0$ no agent has enough evidence for $a$ 
to withstand the $m-\ell$ evidence items for $b$ at $t=1$, so a vote for $a$ is out of the question.
If $k > m-\ell$, then some lucky agents who start out with a lot of evidence for $a$ can still believe $a$ 
even after absorbing all the $b$-evidence. Note that, on average, an agent starts with $k \cdot \nicefrac{\ell}{m}$
items of evidence for $a$ (the expected value of the hypergeometric variable): comparing this to the threshold $m-\ell$ 
gives us the bounds in Lemma \ref{lemma:round-1-agent-accuracies}$\textsf{b}$ and \ref{lemma:round-1-agent-accuracies}$\textsf{c}$
(formally, we use Lemma \ref{lemma:hypergeometric-median-inequality} again).
Note that for $k$ between $\nicefrac{m(m-\ell-1)}{\ell}$ and $\nicefrac{m(m-\ell+1)}{\ell}$ things are less clear: 
this is a transition area with a higher probability of ties, 
where our bounds are not tight enough to determine the outcome with certainty.

The outcome at the end of round $1$ falls out immediately 
from Lemma \ref{lemma:round-1-agent-accuracies}.

\begin{proposition}\label{proposition:round-1-collective-outcome}
    If $\ell > \nicefrac{m}{2}$ and $k$ is odd, it holds that:
    \begin{itemize}
        \item[\textnormal{($\textsf{a}$)}] if $k \leq m-\ell$, 
            then $\lim_{n \rightarrow \infty}{\Pr}\big[ \vote^1(n)=b \big] = 1$
            and deliberation terminates at this round with $b$ as winner; 
        \item[\textnormal{($\textsf{b}$)}] if $m-\ell < k < \nicefrac{m(m-\ell-1)}{\ell}$, 
            then $\lim_{n \rightarrow \infty}{\Pr}\big[ \vote^1(n) = b \big] = 1$.
        \item[\textnormal{($\textsf{c}$)}] if $\nicefrac{m(m-\ell+1)}{\ell} < k < 2(m-\ell)$, 
            then $\lim_{n \rightarrow \infty}{\Pr}\big[ \vote^1(n) = a \big] = 1$, 
            and, additionally, deliberation terminates at this round with $a$ as winner.
    \end{itemize}
\end{proposition}

\noindent
The reasoning here is straightforward: if $k \leq m-\ell$, there is no chance of any vote for $a$,
so $b$ wins unopposed and deliberation ends. If $k$ is between $m-\ell$ and $\nicefrac{m(m-\ell-1)}{\ell}$, 
then at the end of round one the collective decision flips to $b$, and deliberation continues. 
If $k$ is above $\nicefrac{m(m-\ell+1)}{\ell}$, then the collective decision remains $a$ 
(with high probability most agents start with enough evidence for $a$ to remain loyal to $a$ even after absorbing all $b$-evidence), 
the $b$-supporters have exhausted their ammunition and deliberation ends.
This behavior can be seen in Figure \ref{figure:collective-accuracy}, where collective accuracy at $t=1$ starts
picking up only for $k$ above $m-\ell$, and becomes virtually perfect only for $k$ above $\nicefrac{m(m-\ell+1)}{\ell}$.
Note that the transition here is less sharp than the theoretical result suggests:
an effect of using simulations with finite $n$ (Proposition \ref{proposition:round-1-collective-outcome} holds for $n \rightarrow \infty$).

\paragraph{Round two.}
On average, round two takes place if $m-\ell < k < \nicefrac{m(m-\ell-1)}{\ell}$ (Proposition \ref{proposition:round-1-collective-outcome}),
so this is the case we focus on.
We write $p''_a$, $p''_b$ and $p''_\emptyset$ for the probabilities 
of voting for $a$, $b$ and $\emptyset$, respectively, at the end of round $2$, again in the limit as $n$ grows to infinity.
At this point, the group has seen (with high probability) a reversal in the collective decision:
$a$ at $t=0$, then $b$ at $t=1$, induced by the disclosure of all $b$-evidence.
If $k$ is large enough (above $\nicefrac{m(m-\ell-1)}{\ell}$ as per Lemma \ref{lemma:round-1-agent-accuracies}$\textsf{b}$), 
there are still some agents that hold on to their opinion for $a$ at the end of $t=1$ ($p'_a > 0$).
These agents are now in the minority, so they are motivated to disclose for $a$.
By the same argument as in Lemma \ref{lemma:initial-minority-covers-b-evidence},
though applied to $a$-supporters, this minority holds all the evidence for $a$ (as $n \rightarrow \infty$).
Once they disclose it, it is game over.

\begin{lemma}\label{lemma:round-2-agent-accuracies}
    If $\ell > \nicefrac{m}{2}$, $k$ is odd and 
    $m-\ell < k < \nicefrac{m(m-\ell-1)}{\ell}$, 
    then $p''_a = 1$.
\end{lemma}

\noindent
Clearly, once everyone absorbs all the evidence for $a$ there is to know, everyone is convinced that $a$
is the superior alternative (recall that $\ell > \nicefrac{m}{2}$).
The collective decision, then, is flipped back to $a$, and deliberation ends here:
all evidence has been made public, there is nothing more to discuss.

\begin{proposition}\label{proposition:round-2-collective-outcome}
    If $\ell > \nicefrac{m}{2}$, $k$ is odd, and $m-\ell < k < \nicefrac{m(m-\ell-1)}{\ell}$,
    then $\lim_{n \rightarrow \infty}{\Pr}\big[ \vote^2(n) = a \big] = 1$
    and deliberation terminates at this round with $a$ as winner.
\end{proposition}

\noindent
The collective decision at the second round is illustrated in Figure \ref{figure:collective-accuracy}, 
where we see an increase in collective accuracy at $t=2$ for $k$ 
in between $m-\ell$ and $\nicefrac{m(m-\ell-1)}{\ell}$.
Again, the transition here is not drastic: the reason, as with the case for $t=1$, is the small $n$
used in the simulations (the transitions become sharper as $n$ grows).

Our jury theorem is a matter of putting together the results obtained
so far in a statement about the outcome of deliberation at termination.

\begin{theorem}\label{theorem:jury-theorem-deliberation}
    For $\ell > \nicefrac{m}{2}$ and odd $k$, 
    it holds that
    if $m-\ell < k < \nicefrac{m(m-\ell-1)}{\ell}$ 
    or $\nicefrac{m(m-\ell+1)}{\ell} < k$,
    then $\lim_{n \rightarrow \infty}{\Pr}\big[ \vote^*(n) = a \big] = 1$.
    If $k \leq m-\ell$, then $\lim_{n \rightarrow \infty}{\Pr}\big[ \vote^*(n) = b \big] = 1$.
    Additionally, deliberation terminates, with high probability, after at most two rounds.
\end{theorem}

\section{Biased Sampling}\label{section:biased-sampling}

\begin{figure}[t]
    \centering
    \begin{subfigure}[b]{0.49\textwidth}
        \includegraphics[width=\textwidth]{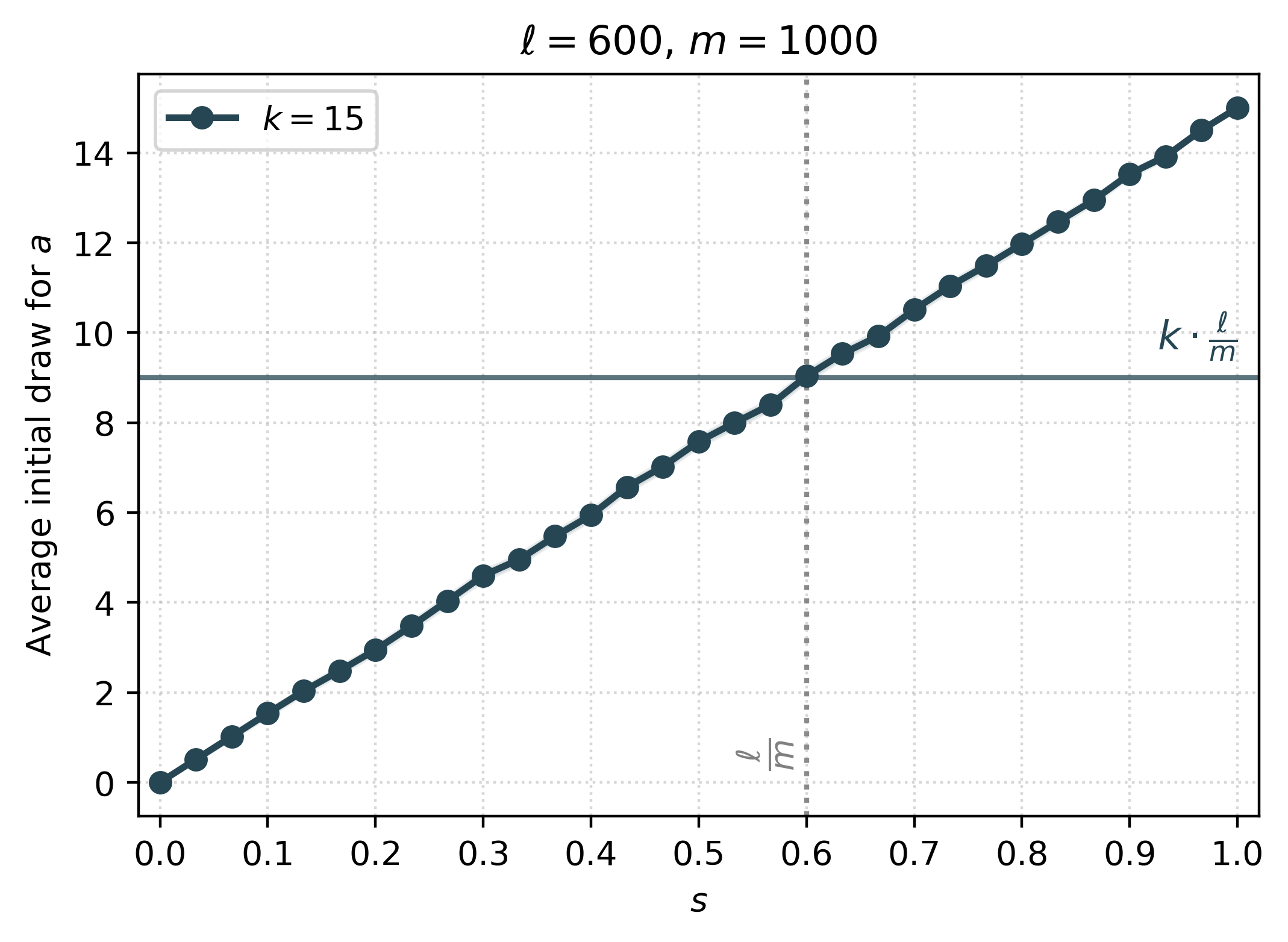}
        \caption{%
            Average number of evidence items for $a$ in an agent's initial draw.
            }
        \label{figure:salience-avg-successes}
    \end{subfigure}
    \begin{subfigure}[b]{0.49\textwidth}
        \includegraphics[width=\textwidth]{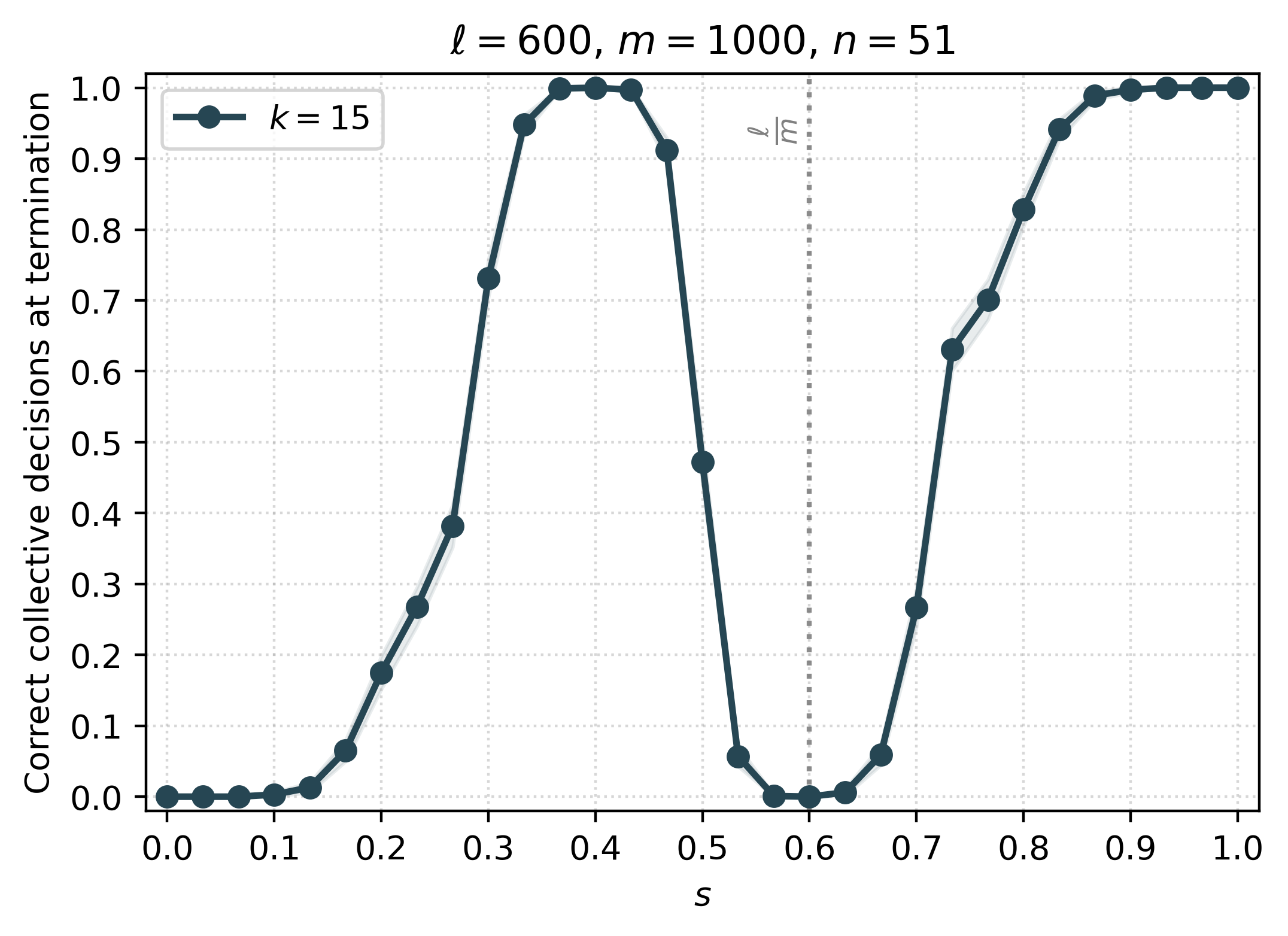}
        \caption{%
            Average collective accuracy at the end of deliberation.
            }
        \label{figure:salience-collective-accuracy}
    \end{subfigure}
    \begin{subfigure}[b]{0.49\textwidth}
        \includegraphics[width=\textwidth]{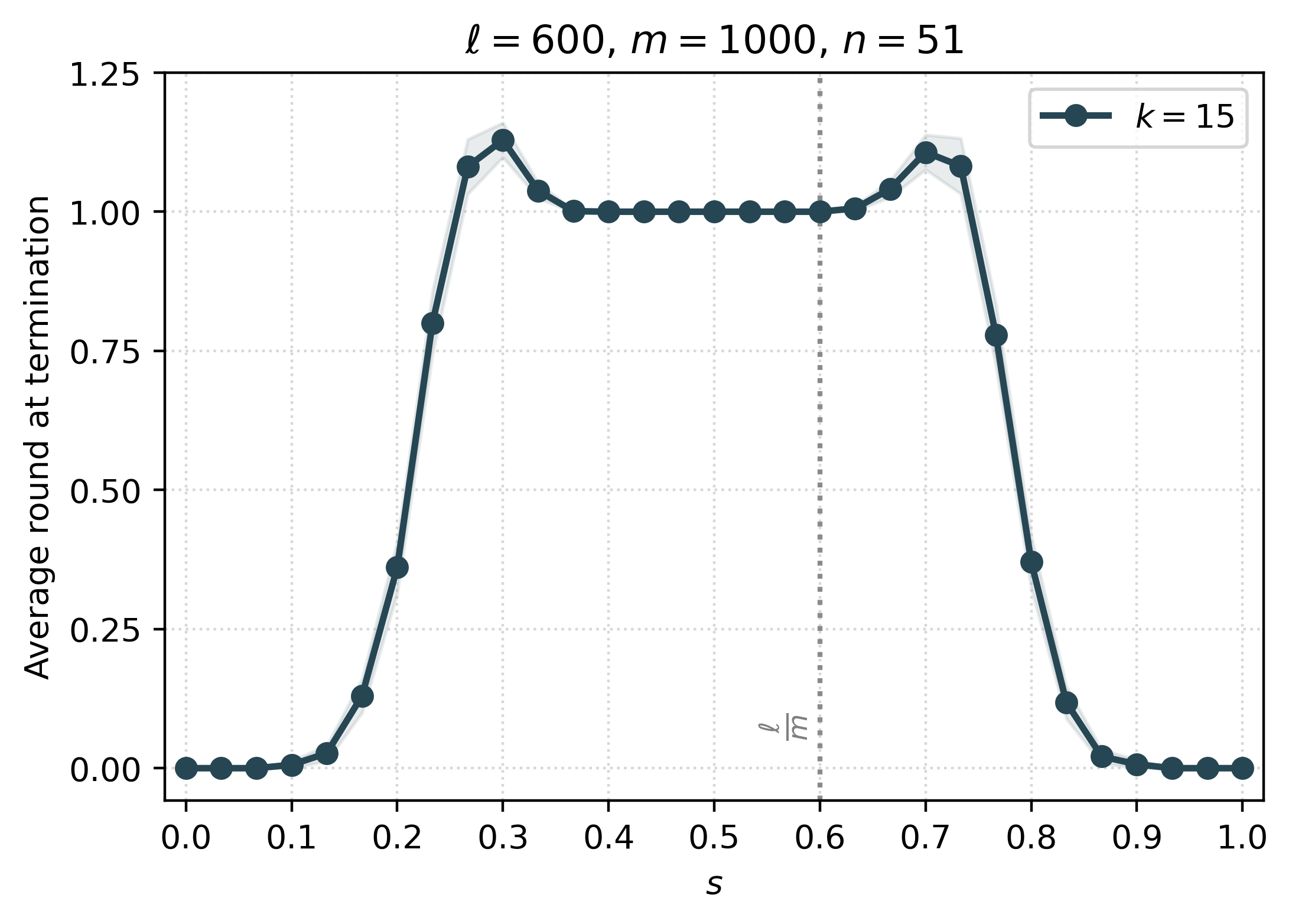}
        \caption{%
            Average round at which deliberation terminates.
            }
        \label{figure:salience-avg-nr-rounds}
    \end{subfigure}
    \begin{subfigure}[b]{0.49\textwidth}
        \includegraphics[width=\textwidth]{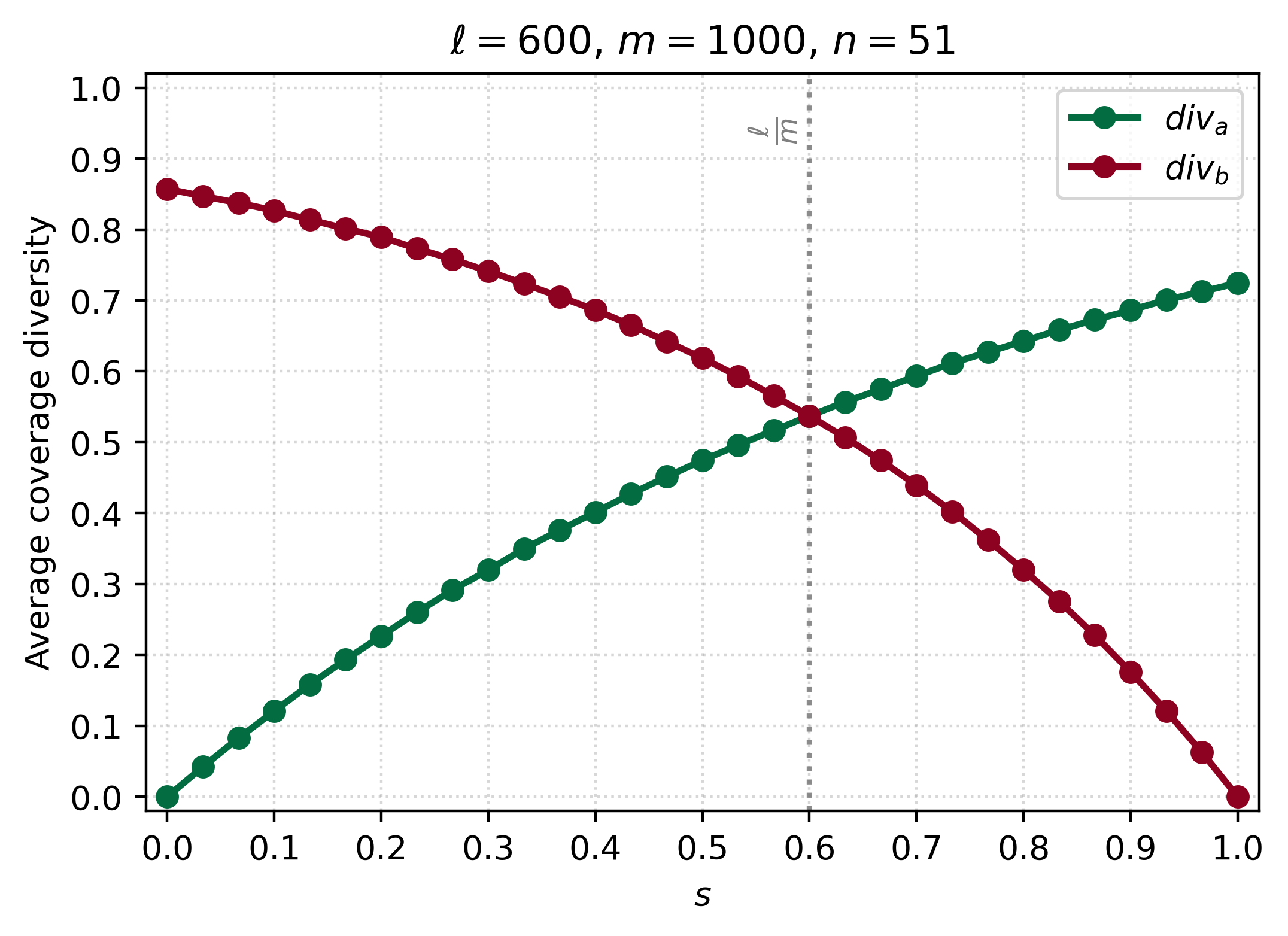}
        \caption{%
            Average coverage diversity of the initial evidence draws.
            }
        \label{figure:salience-diversity}
    \end{subfigure}
    \caption{%
        Simulation results for biased sampling with salience factor $s$.
        Each data point is an average across 1000 simulations.
    }
    \label{figure:salience-results}
\end{figure}

If $k$ is small-ish (below $m-\ell$), 
Theorem \ref{theorem:jury-theorem-deliberation} says that 
deliberation is likely to do badly. Can we do better for such values of $k$?

A natural way to generalize our model is to add a notion of salience 
in the initial acquisition of evidence.
Intuitively, more salient items of evidence should be sampled more often: perhaps
because they are more accesible, or more likely to be mentioned in conversation
(e.g., in a competition between two products one company may invest 
more in advertising, making the virtues of its product more visible).
Formally, we model this with a \emph{salience factor} $s \in [0,1]$
that gets normalized to a probability distribution $
    \bm{s} = (
        \underbrace{s_a, \dots, s_a}_{\ell}, 
        \underbrace{s_b, \dots, s_b}_{m-\ell}
        )
$
over the $m$ items of evidence,
where $s_a = \nicefrac{s}{\ell}$ and $s_b = \nicefrac{1-s}{m-\ell}$
are the probabilities of sampling items of evidence for $a$ and $b$, respectively.
The probabilities are chosen such that evidence items supporting $a$ (respectively, $b$) are sampled
with probability proportional to $s$ (respectively, $1-s$).
For $s = \nicefrac{\ell}{m}$ we get the baseline case where each item of evidence has the same salience.
This results in \emph{biased sampling}, a model described 
by the Wallenius noncentral hypergeometric distribution
\cite{chesson_non-central_1976,fog_calculation_2008,wallenius_biased_1973}.
Here we just illustrate its effect on the initial draws in Figure \ref{figure:salience-avg-successes}.
Note that: ($i$) as $s$ increases
so does the average number of evidence items for $a$ in an agent's initial draw, and 
($ii$) for $s = \nicefrac{\ell}{m}$ the average number of evidence items 
for $a$ is very well approximated by 
$k \cdot\nicefrac{\ell}{m}$, the expected value of the (unbiased) hypergeometric distribution.

With biased sampling, the average collective accuracy at the end of 
deliberation (Figure \ref{figure:salience-collective-accuracy})
shows a striking, double-peaked pattern: 
for $s$ around $\nicefrac{\ell}{m}$ (equal saliences) accuracy is close to 0, as
predicted by Theorem \ref{theorem:jury-theorem-deliberation} for $k$ below $m-\ell$;
but for $s$ on either side of this value, accuracy increases dramatically. What is going on?
The explanation is that for $s$ smaller than $0.2$, or larger than $0.8$, 
evidence for $a$ is either not at all salient,
or so salient that it is almost always sampled, 
i.e., deliberation is likely to terminate immediately with a unanimous 
vote for either $b$ or $a$, respectively. 
The interesting case is $s$ around $[0.3, 0.5]$, where group performance is good
even though $a$ is less salient than $b$:
$a$, here, turns out to be salient enough to have supporters, though not a majority.
This, now, invites disclosure for $a$ at $t=1$, which is very likely decisive.
This explanation is supported by Figure \ref{figure:salience-avg-nr-rounds}, 
which shows that, for $s$ in this range, deliberation usually terminates at round $1$.

Note that what moves the needle in a deliberation round is not 
just having many agents with evidence: 
it is that, taken together, agents cover parts of the evidence space that others do not.
This intuition underlies the generic idea that diversity 
boosts the wisdom of crowds \cite{Hong2004u,Page2008k,Landemore2013z}.
We can explore this here by defining a simple measure of \emph{coverage diversity}, 
which we define as the fraction of the total evidence items 
for alternative in the initial draws of the agents:
$$
    \textit{div}_a = \frac{1}{\ell} \cdot \left| \bigcup_{i \in N} A_i^0 \right|, 
    \quad
    \textit{div}_b = \frac{1}{m-\ell} \cdot \left| \bigcup_{i \in N} B_i^0 \right|. 
$$
Figure \ref{figure:salience-diversity} shows how coverage diversity varies with $s$: as $s$ increases 
(evidence for $a$ becomes more salient), 
the profile of initial draws covers more of the total evidence for $a$. 
The opposite effect happens for $b$ and, reassuringly, 
coverage diversity is equal for $s = \nicefrac{\ell}{m}$.
Note, however, that diversity is a double edged sword:
high diversity for only one of the alternatives leads to asymmetric outcomes 
(one side hogging all the evidence); there is likely a sweet spot in the diversity 
of an evidence profile that is useful for deliberation, and worth more investigation.

\section{Discussion}\label{section:conclusion}
We believe formal models of deliberation along the lines 
discussed here can inform the design of discussion-based decision processes. 
As an example, business owners can reply to reviews on Google Maps, 
but users cannot reply to these replies: discussion is effectively capped at the first round,
which, we have seen, can prevent disclosure of potentially relevant information. 
By contrast, reviewing rounds in most CS conferences allow as many replies as necessary 
in the discussion phase between reviewers (as well as rebuttals from the authors),
which, in light of Theorem \ref{theorem:jury-theorem-deliberation}, seems like a better design choice.

Theorem \ref{theorem:jury-theorem-deliberation} seems encouraging (there is a broad range 
of $k$'s for which deliberation leads to the correct outcome!).
On closer look, though, the range of values for which deliberation does badly ($k$ below $m-\ell$) 
is perhaps more realistic (see Figure \ref{figure:collective-accuracy}):
essentially, if $\ell$ is not much larger than $m-\ell$ 
(the superior alternative does not have a huge advantage over its competitor),
then $k$ has to be rather large (agents have to be diligent and explore a lot 
before forming their first opinion)
for deliberation not to mess things up. Not as encouraging, perhaps.
But Section \ref{section:biased-sampling} shows that, even for small $k$, 
biased sampling can mitigate the problem by, counterintuitively, 
making the superior alternative a bit less salient.

A clear limitation of our model is that it leaves out elements of real-world deliberation 
that are likely to be important (strategic behavior, the ability to reject evidence).
We think a lot of insight can be offered even in a stylized model 
but future work will seek to make the model cover more ground.


%
%
%
\bibliographystyle{splncs04}
\bibliography{bibliography.bib}

@book{Bachtiger2024s,
  title     = {{Deliberative Democracy for Diabolical Times: Confronting
               Populism, Extremism, Denial, and Authoritarianism}},
  author    = {Bächtiger, André and Dryzek, John S},
  publisher = {Cambridge University Press},
  year      = 2024
}

@book{Brennan2016t,
  title     = {{Against Democracy}},
  author    = {Brennan, Jason},
  publisher = {Princeton University Press},
  year      = 2016
}

@article{Caluwaerts2016r,
  title     = {{Generating Democratic Legitimacy through Deliberative
               Innovations: The Role of Embeddedness and Disruptiveness}},
  author    = {Caluwaerts, Didier and Reuchamps, Min},
  journal   = {Representation},
  publisher = {Informa UK Limited},
  volume    = 52,
  number    = 1,
  pages     = {13--27},
  year      = 2016
}

@article{Caluwaerts2023l,
  title     = {{Deliberation and polarization: a multi-disciplinary review}},
  author    = {Caluwaerts, Didier and Bernaerts, Kamil and Kesberg, Rebekka and
               Smets, Lien and Spruyt, Bram},
  journal   = {Frontiers in Political Science},
  publisher = {Frontiers Media SA},
  volume    = 5,
  number    = 1127372,
  pages     = 1127372,
  year      = 2023
}

@book{Caplan2008w,
  title     = {{The Myth of the Rational Voter: Why Democracies Choose Bad Policies}},
  author    = {Caplan, Bryan},
  publisher = {Princeton University Press},
  address   = {Princeton, NJ},
  year      = 2008
}

@inproceedings{Chingoma2023i,
  title     = {{Deliberation as Evidence Disclosure: A Tale of Two Protocol Types}},
  author    = {Chingoma, Julian and Haret, Adrian},
  booktitle = {{Proceedings of the Thirty-Second International Joint Conference on Artificial Intelligence (IJCAI 2023)}},
  pages     = {2589--2597},
  year      = 2023
}

@incollection{Cohen1997h,
  title     = {{Deliberation and Democratic Legitimacy}},
  author    = {Cohen, Joshua},
  booktitle = {{Deliberative Democracy}},
  publisher = {The MIT Press},
  year      = 1997
}

@article{Dietrich2025k,
  title     = {{Deliberation and the wisdom of crowds}},
  author    = {Dietrich, Franz and Spiekermann, Kai},
  journal   = {Economic Theory},
  publisher = {Springer},
  volume    = 79,
  number    = 2,
  pages     = {603--655},
  year      = 2025
}

@article{Ding2021d,
  title     = {{Deliberation and epistemic democracy}},
  author    = {Ding, Huihui and Pivato, Marcus},
  journal   = {Journal of Economic Behavior \& Organization},
  publisher = {Elsevier},
  volume    = 185,
  pages     = {138--167},
  year      = 2021
}

@article{Dryzek2001m,
  title     = {{Legitimacy and Economy in Deliberative Democracy}},
  author    = {Dryzek, John S.},
  journal   = {Political Theory},
  publisher = {SAGE Publications},
  volume    = 29,
  number    = 5,
  pages     = {651--669},
  year      = 2001
}

@article{Dryzek2019q,
  title     = {{The crisis of democracy and the science of deliberation}},
  author    = {Dryzek, John S. and Bächtiger, André and Chambers, Simone and
               Cohen, Joshua and Druckman, James N. and Felicetti, Andrea and
               Fishkin, James S. and Farrell, David M. and Fung, Archon and
               Gutmann, Amy and Landemore, Hélène and Mansbridge, Jane and
               Marien, Sofie and Neblo, Michael A and Niemeyer, Simon and
               Setälä, Maija and Slothuus, Rune and Suiter, Jane and Thompson,
               Dennis and Warren, Mark E.},
  journal   = {Science},
  publisher = {American Association for the Advancement of Science (AAAS)},
  volume    = 363,
  number    = 6432,
  pages     = {1144--1146},
  year      = 2019
}

@book{Estlund2008a,
  title     = {{Democratic Authority: A Philosophical Framework}},
  author    = {Estlund, David M.},
  publisher = {Princeton University Press},
  address   = {Princeton, NJ},
  year      = 2008
}

@incollection{Estlund2018r,
  title     = {{The Epistemic Value of Democratic Deliberation}},
  author    = {Estlund, David and Landemore, Hélène},
  editor    = {Bächtiger, Andre and Dryzek, John S. and Mansbridge, Jane and Warren, Mark E.},
  booktitle = {{The Oxford Handbook of Deliberative Democracy}},
  publisher = {Oxford University Press},
  pages     = {113--131},
  year      = 2018
}

@article{Farrar2010r,
  title     = {{Disaggregating Deliberation’s Effects: 
               An Experiment within a Deliberative Poll}},
  author    = {Farrar, Cynthia and Fishkin, James S. and Green, Donald P. and
               List, Christian and Luskin, Robert C. and Levy Paluck, Elizabeth},
  journal   = {British Journal of Political Science},
  publisher = {Cambridge University Press (CUP)},
  volume    = 40,
  number    = 2,
  pages     = {333--347},
  year      = 2010
}

@book{Fishkin2025f,
  title     = {{Can Deliberation Cure the Ills of Democracy?}},
  author    = {Fishkin, James S.},
  publisher = {Oxford University Press},
  address   = {London, England},
  year      = 2025
}

@article{Germann2024h,
  title     = {{Scaling up? Unpacking the effect of deliberative mini-publics on
               legitimacy perceptions}},
  author    = {Germann, Micha and Marien, Sofie and Muradova, Lala},
  journal   = {Political Studies},
  publisher = {SAGE Publications},
  volume    = 72,
  number    = 2,
  pages     = {677--700},
  year      = 2024
}

@article{Goodin2006r,
  title     = {{Deliberative Impacts: The Macro-Political Uptake of
               Mini-Publics}},
  author    = {Goodin, Robert E and Dryzek, John S},
  journal   = {Politics \& Society},
  publisher = {SAGE Publications},
  volume    = 34,
  number    = 2,
  pages     = {219--244},
  year      = 2006
}

@article{Gronlund2015r,
  title     = {{Does Enclave Deliberation Polarize Opinions?}},
  author    = {Grönlund, Kimmo and Herne, Kaisa and Setälä, Maija},
  journal   = {Political Behavior},
  publisher = {Springer Science and Business Media LLC},
  volume    = 37,
  number    = 4,
  pages     = {995--1020},
  year      = 2015
}

@article{Hong2004u,
  title     = {{Groups of diverse problem solvers can outperform groups of
               high-ability problem solvers}},
  author    = {Hong, Lu and Page, Scott E.},
  journal   = {Proceedings of the National Academy of Sciences of the United
               States of America},
  publisher = {National Acad Sciences},
  volume    = 101,
  number    = 46,
  pages     = {16385--16389},
  year      = 2004
}

@book{Landemore2013z,
  title     = {{Democratic Reason: Politics, Collective Intelligence, and the
               Rule of the Many}},
  author    = {Landemore, Hélène},
  publisher = {Princeton University Press},
  year      = 2013
}

@article{Landemore2014p,
  title     = {{Democracy as Heuristic: The Ecological Rationality of Political Equality}},
  author    = {Landemore, Hélène},
  journal   = {The Good Society},
  publisher = {The Pennsylvania State University Press},
  volume    = 23,
  number    = 2,
  pages     = {160--178},
  year      = 2014
}

@book{Lehmann2022a,
  title     = {{Testing Statistical Hypotheses}},
  author    = {Lehmann, E. L. and Romano, Joseph P.},
  publisher = {Springer},
  address   = {Cham, Switzerland},
  edition   = 4,
  series    = {Springer Texts in Statistics},
  year      = 2022
}

@article{List2011j,
  title     = {{Group Communication and the Transformation of Judgments: An Impossibility Result}},
  author    = {List, Christian},
  journal   = {The Journal of Political Philosophy},
  publisher = {Wiley},
  volume    = 19,
  number    = 1,
  pages     = {1--27},
  year      = 2011
}

@article{List2013m,
  title     = {{Deliberation, Single-Peakedness, and the Possibility of
               Meaningful Democracy: Evidence from Deliberative Polls}},
  author    = {List, Christian and Luskin, Robert C. and Fishkin, James S and
               McLean, Iain},
  journal   = {The Journal of Politics},
  publisher = {University of Chicago Press},
  volume    = 75,
  number    = 1,
  pages     = {80--95},
  year      = 2013
}

@incollection{Mendelberg2002p,
  title     = {{The Deliberative Citizen: Theory and Evidence}},
  author    = {Mendelberg, Tali},
  booktitle = {{Political Decision Making, Deliberation and Participation}},
  publisher = {Elsevier},
  volume    = 6,
  pages     = {151--193},
  year      = 2002
}

@article{Miller1992d,
  title     = {{Deliberative Democracy and Social Choice}},
  author    = {Miller, David},
  journal   = {Political Studies},
  publisher = {SAGE Publications},
  volume    = 40,
  number    = {1},
  pages     = {54--67},
  year      = 1992
}

@article{Siebe2024j,
  title     = {{The interdependence of social deliberation and judgment aggregation}},
  author    = {Siebe, Hendrik},
  journal   = {Social Choice and Welfare},
  publisher = {Springer Science and Business Media LLC},
  volume    = 63,
  number    = {3-4},
  pages     = {475--507},
  year      = 2024
}

@article{Stasser1985q,
  title     = {{Pooling of unshared information in group decision making: 
               Biased information sampling during discussion}},
  author    = {Stasser, Garold and Titus, William},
  journal   = {Journal of Personality and Social Psychology},
  publisher = {American Psychological Association (APA)},
  volume    = 48,
  number    = 6,
  pages     = {1467--1478},
  year      = 1985
}

@article{Sunstein2002m,
  title     = {{The Law of Group Polarization}},
  author    = {Sunstein, Cass R.},
  journal   = {The Journal of Political Philosophy},
  publisher = {Wiley},
  volume    = 10,
  number    = 2,
  pages     = {175--195},
  year      = 2002
}

@book{Sunstein2014h,
  title     = {{Wiser: Getting Beyond Groupthink to Make Groups Smarter}},
  author    = {Sunstein, Cass R. and Hastie, Reid},
  publisher = {Harvard Business Review Press},
  address   = {Boston, MA},
  year      = 2014
}

@BOOK{Goodin2018r,
  title     = "{An Epistemic Theory of Democracy}",
  author    = "Goodin, Robert E. and Spiekermann, Kai",
  publisher = "Oxford University Press",
  year      =  2018
}

@ARTICLE{List2001g,
  title     = "{Epistemic Democracy: Generalizing the Condorcet Jury Theorem}",
  author    = "List, Christian and Goodin, Robert E.",
  journal   = "The Journal of Political Philosophy",
  publisher = "Blackwell",
  volume    =  9,
  number    =  3,
  pages     = "277--306",
  year      =  2001
}

@INCOLLECTION{Dietrich2023h,
  title     = "{Jury Theorems}",
  author    = "Dietrich, Franz and Spiekermann, Kai",
  editor    = "Zalta, Edward N. and Nodelman, Uri",
  booktitle = "{The Stanford Encyclopedia of Philosophy}",
  publisher = "Metaphysics Research Lab, Stanford University",
  edition   = "Spring 2023",
  year      =  2023
}

@ARTICLE{Grofman1983e,
  title   = "{Thirteen theorems in search of the truth}",
  author  = "Grofman, Bernard and Owen, Guillermo and Feld, Scott L.",
  journal = "Theory and Decision",
  volume  =  15,
  number  =  3,
  pages   = "261--278",
  year    =  1983
}

@INPROCEEDINGS{Goel2025d,
  title     = "{Metric distortion of small-group deliberation}",
  author    = "Goel, Ashish and Goyal, Mohak and Munagala, Kamesh",
  booktitle = "{Proceedings of the 57th Annual ACM Symposium on Theory of
               Computing}",
  pages     = "1568--1579",
  year      =  2025
}

@ARTICLE{Goyal2026l,
  title         = "{Beyond polarization: Opinion mixing and social influence in
                   deliberation}",
  author        = "Goyal, Mohak and Gelauff, Lodewijk and Gupta, Naman and Goel,
                   Ashish and Munagala, Kamesh",
  journal       = "arXiv [cs.SI]",
  year          =  2026,
  archivePrefix = "arXiv",
  primaryClass  = "cs.SI"
}

@BOOK{Page2008k,
  title     = "{The Difference: How the Power of Diversity Creates Better
               Groups, Firms, Schools, and Societies}",
  author    = "Page, Scott E.",
  publisher = "Princeton University Press",
  address   = "Princeton, NJ",
  year      =  2008
}

@ARTICLE{Luan2012x,
  title     = "{When Does Diversity Trump Ability (and Vice Versa) in Group
               Decision Making? A Simulation Study}",
  author    = "Luan, Shenghua and Katsikopoulos, Konstantinos V. and Reimer,
               Torsten",
  journal   = "PloS One",
  publisher = "Public Library of Science (PLoS)",
  volume    =  7,
  number    =  2,
  pages     = "e31043",
  year      =  2012
}

@ARTICLE{Reijula2021y,
  title     = "{The Diversity-Ability Trade-Off in Scientific Problem Solving}",
  author    = "Reijula, Samuli and Kuorikoski, Jaakko",
  journal   = "Philosophy of Science",
  publisher = "Cambridge University Press (CUP)",
  volume    =  88,
  number    =  5,
  pages     = "894--905",
  year      =  2021
}

@article{chesson_non-central_1976,
	title = {A {Non}-{Central} {Multivariate} {Hypergeometric} {Distribution} {Arising} from {Biased} {Sampling} with {Application} to {Selective} {Predation}},
	volume = {13},
	issn = {0021-9002},
	number = {4},
	urldate = {2026-05-26},
	journal = {Journal of Applied Probability},
	publisher = {Applied Probability Trust},
	author = {Chesson, Jean},
	year = {1976},
	pages = {795--797},
}

@article{fog_calculation_2008,
	title = {Calculation {Methods} for {Wallenius}' {Noncentral} {Hypergeometric} {Distribution}},
	volume = {37},
	issn = {0361-0918},
	number = {2},
	urldate = {2026-05-26},
	journal = {Communications in Statistics - Simulation and Computation},
	publisher = {Taylor \& Francis},
	author = {Fog, Agner},
	year = {2008},
	pages = {258--273},
}

@phdthesis{wallenius_biased_1973,
	title = {Biased {Sampling}; {The} {Noncentral} {Hypergeometric} {Probability} {Distribution}},
	url = {https://purl.stanford.edu/wh056vj9347},
	language = {en},
	urldate = {2026-05-26},
	author = {Wallenius, K.T.},
	year = {1973},
}

@article{cooper_linearly_2010,
	title = {Linearly bounded liars, adaptive covering codes, and deterministic random walks},
	volume = {1},
	issn = {21563527, 2150959X},
	language = {en},
	number = {3},
	urldate = {2026-05-12},
	journal = {Journal of Combinatorics},
	author = {Cooper, Joshua N.},
	year = {2010},
	pages = {307--333}
}

@article{siegel_median_2001,
	title = {Median {Bounds} and {Their} {Application}},
	volume = {38},
	copyright = {https://www.elsevier.com/tdm/userlicense/1.0/},
	issn = {01966774},
	language = {en},
	number = {1},
	urldate = {2026-05-12},
	journal = {Journal of Algorithms},
	author = {Siegel, Alan},
	year = {2001},
	pages = {184--236},
}

\newpage
\section*{Appendix}
Lemma \ref{lemma:hypergeometric-median-inequality},
stating that the median of the hypergeoemetric distribution, 
is between the floor and the ceiling of the mean,
appears as Corollary 13 of Theorem 12 in \cite{cooper_linearly_2010}, 
which relies, in turn, on Theorem 2.4 in \cite{siegel_median_2001}.

\begin{proof}[Lemma \ref{lemma:round-0-agent-accuracies}]
    ($\textsf{a}$) 
    If $k \geq 2(m-\ell)$ it becomes impossible to draw more than $\nicefrac{k}{2}$
    evidence items for $b$, since there are only $m-\ell$ such items in total. 
    Hence, at $t=0$ agents always sample at least $\nicefrac{k}{2}$ evidence items for $a$, 
    and end up either voting for $a$ or being indifferent,
    and never vote for $b$: $p_a > 0$, $0 \leq p_\emptyset < 1$
    ($p_\emptyset$ cannot be $1$ because $p_a > 0$) and $p_b = 0$.
    Strengthening the condition to $k > 2(m-\ell)$ (or making $k$ odd)
    rules out ties in evidence amounts and makes things more clear-cut:
    $p_\emptyset = 0$ and $p_a = 1$.

    ($\textsf{b}$) 
    If $k < 2(m-\ell)$ it becomes possible to end up voting for $b$ at $t=0$
    (e.g., by drawing all $k$ evidence items for $b$), so $p_b > 0$.
    This is assumed to hold for the remaining sub-cases.

    ($\textsf{b}_1$) 
    If $k$ is odd there is no possibility for ties at $t=0$, so $p_\emptyset=0$. 
    Thus, the events $|A_i^0| > \nicefrac{k}{2}$ and $|A_i^0| < \nicefrac{k}{2}$ are complementary:
    \begin{align}\label{equation:complementary-events}
        p_a + p_b = 1.
    \end{align}
    Both the correct and incorrect vote can be further broken down into $\lfloor \nicefrac{k}{2} \rfloor$ mutually disjoint events of the form $|A_i^0| = j$, for $j \in [k]$, such that:
    \begin{align*}
        p_a 
        = {\Pr}{\left[ |A_i^0| > \frac{k}{2} \right]} & 
        = \sum^k_{j = \lfloor \nicefrac{k}{2}\rfloor + 1} {\Pr}\Big[ |A_i^0| = j \Big], ~\text{and}\\
        p_b
        = {\Pr}{\left[ |A_i^0| < \frac{k}{2} \right]} & 
        = \sum^{\lfloor \nicefrac{k}{2}\rfloor}_{j = 0} {\Pr}\Big[ |A_i^0| = j \Big]
        = \sum^k_{j = \lfloor \nicefrac{k}{2}\rfloor + 1} {\Pr}\Big[ |A_i^0| = k-j \Big].
    \end{align*}
    We now show that each term in the first sum is greater than the corresponding term in the second sum.
    Taking $j > \nicefrac{k}{2}$, we have:
    \begin{align*}
        {\Pr}\Big[ |A_i^0| = j \Big] > {\Pr}\Big[ |A_i^0| = k-j \Big] & 
        \text{ iff } \frac{\binom{\ell}{j}\binom{m-\ell}{k-j}}{\binom{m}{k}} > 
        \frac{\binom{\ell}{k-j}\binom{m-\ell}{j}}{\binom{m}{k}} \\ 
        & \text{ iff } \frac{\binom{\ell}{j}}{\binom{\ell}{k-j}} > \frac{\binom{m-\ell}{j}}{\binom{m-\ell}{k-j}}\\
        & \text{ iff } \frac{(\ell-k+j)!}{(\ell-j)!} > \frac{(m-\ell-k+j)!}{(m-\ell-j)!}\\
        & \text{ iff } \prod_{s=1}^{2j-k} (\ell - j + s) > 
            \prod_{s=1}^{2j-k} (m-\ell - j + s).
    \end{align*}
    With $j > \nicefrac{k}{2}$ we get that $2j-k > 0$,
    and thus the products on the last line are well-defined, with $2j-k$ terms each. 
    What is more, $\ell > \nicefrac{m}{2}$ gives us that 
    $\ell - j + s > m-\ell - j + s$ for all $s \in \{1, \dots, 2j-k\}$.
    Thus, it follows that ${\Pr}\big[ |A_i^0| = j \big] > {\Pr}\big[ |A_i^0| = k-j \big]$
    for all $j > \nicefrac{k}{2}$, and hence:
    \begin{align*}
        p_a & =
        \sum^k_{j = \lfloor \nicefrac{k}{2}\rfloor + 1} {\Pr}\Big[ |A_i^0| = j \Big] \\
        & > \sum^k_{j = \lfloor \nicefrac{k}{2}\rfloor + 1} {\Pr}\Big[ |A_i^0| = k-j \Big] \\ 
        & = p_b.
    \end{align*}
    With $p_a > p_b$ and $p_a + p_b = 1$ (Equation \ref{equation:complementary-events})
    we conclude that $p_a > \nicefrac{1}{2}$. 

    If $k$ is even this argument no longer works:
    there is the additional possibility of a tie ($|A_i^0| = \nicefrac{k}{2}$),
    which means that $p_\emptyset > 0$ and Equation \ref{equation:complementary-events} does not hold.
    The argument from ($\textsf{b}_1$) only gives us that $p_a > p_b$, 
    but this is now insufficient to establish that $p_a > \nicefrac{1}{2}$ 
    (we would need $p_a > p_b + p_\emptyset$ for that)
    and, as a matter of fact, is not even true for the same conditions as in ($\textsf{b}_1$).
    For $m=30$, $\ell = 16$ and $k=2$, we get $p_a \approx 0.28$,
    well below $\nicefrac{1}{2}$. In other words, for even $k$ we need stronger conditions
    to ensure that initial accuracy is above $\nicefrac{1}{2}$. 
    A set of sufficient conditions is offered in ($\textsf{b}_2$).

    ($\textsf{b}_2$) 
    The first inequality of Lemma \ref{lemma:hypergeometric-median-inequality}
    ($\lfloor \mu \rfloor \leq M$), together with the definition of median $M$,
    implies that:
    $$
        {\Pr}\Big[ |A_i^0| \geq \lfloor \mu \rfloor \Big] \geq 
        {\Pr}\Big[ |A_i^0| \geq M \Big] \geq \frac{1}{2},
    $$
    where $\mu = k \cdot \nicefrac{\ell}{m}$ is the expected value of the 
    hypergeometric variable $|A_i^0|$ and $M$ is its median.
    Thus, to ensure that the probability of a correct initial vote is above $\nicefrac{1}{2}$, 
    it suffices to push $\nicefrac{k}{2}$ below $\lfloor \mu \rfloor$---and the precondition of 
    claim ($iii$) is formulated to ensure precisely this. More precisely, note
    that $k > (\nicefrac{\ell}{m} - \nicefrac{1}{2})^{-1}$ is equivalent to:
    $$
        \mu = k\frac{\ell}{m} > \frac{k}{2} + 1,
    $$
    which, in turn, implies that $\lfloor \mu \rfloor \geq \nicefrac{k}{2} + 1$.
    Chaining this with the inequality $M \geq \lfloor \mu \rfloor$ 
    from Lemma \ref{lemma:hypergeometric-median-inequality} 
    we get that $M \geq \nicefrac{k}{2}+1$, and thus:
    $$
        p_a =
        {\Pr}{\left[ |A_i^0| > \frac{k}{2} \right]} = 
        {\Pr}{\left[ |A_i^0| \geq \frac{k}{2}+1 \right]} \geq 
        {\Pr}\Big[ |A_i^0| \geq M \Big] \geq 
        \frac{1}{2}.
    $$
    Why do we also need to have $\nicefrac{\ell}{m} > \nicefrac{1}{2} + \nicefrac{1}{m}$?
    We add it to ensure that the lower bound on $k$ we found is consistent with the assumption
    that $k \leq m$.
    If, on the contrary, it would hold that 
    $\nicefrac{\ell}{m} \leq \nicefrac{1}{2} + \nicefrac{1}{m}$, 
    then $k$ would have to satisfy: 
    $$
        k > \left(\frac{\ell}{m} - \frac{1}{2}\right)^{-1} \geq \left(\frac{1}{m}\right)^{-1} = m,
    $$
    which we do not allow for.
\end{proof}

\begin{proof}[Proposition \ref{proposition:round-0-collective-outcome}]
    ($\textsf{a}$)
    If $k \geq 2(m-\ell)$, then, 
    by Lemma \ref{lemma:round-0-agent-accuracies},
    every agent either votes for $a$ or is indifferent and no one votes for $b$.
    This does not mean that the collective vote at $t=0$ is necessarily $a$: 
    in the edge case where everyone is indifferent, the collective vote is tied.
    However, as $n$ grows, the probability that everyone is indifferent goes to $0$,
    so the probability that $\vote^0(n) = a$ goes to $1$.
    Moreover, since there are no dissenters at $t=0$, deliberation terminates immediately.
    Note that if we strengthen the condition to $k > 2(m-\ell)$, 
    there can no longer be any indifferent voters, and deliberation terminates with
    $a$ as the unanimous winner, for any $n$. 

    ($\textsf{b}$)
    If $k$ is odd there are no indifferent voters at $t=0$ ($p_\emptyset = 0$)
    and each agent votes for $a$ with probability $p > \nicefrac{1}{2}$
    (Lemma \ref{lemma:round-0-agent-accuracies}).  
    By the Law of Large Numbers, the probability of a correct majority vote 
    grows asymptotically to $1$ as $n$ goes to infinity.
\end{proof}

\begin{proof}[Lemma \ref{lemma:initial-minority-covers-b-evidence}]
    Consider the complementary event, i.e., 
    that there is at least one evidence item for $b$
    that is \emph{not} in any of the $B_i^0$, for $i \in N_b$.
    Conditioning on the possible values of $n_b$, we can write:
    $$
        {\Pr}{\left[ \bigcup_{i \in N_b} B^0_i \neq B \right]} = 
        \sum^n_{s=0} {\Pr}{\left[ \bigcup_{i \in N_b} B^0_i \neq B \mid n_b = s \right]} \cdot 
        {\Pr}{\left[ n_b = s \right]}.
    $$
    We show that the probability of this event goes to $0$ as $n$ grows.

    Note, first, that each voter $i \in N$ votes for $b$ with probability $1-p$.
    Hence, the probability that $n_b$ takes on a specific value $s \in \{0, \dots, n\}$ 
    is given by the binomial distribution with parameters $n$ and $1-p$:
    $$
        {\Pr}{\left[ n_b = s \right]} = {n \choose s} (1-p)^s p^{n-s}.
    $$
    Second, for a fixed $n_b = s$, the probability that the union of the $B_i^0$'s 
    does not cover $B$ can be bounded as follows:
    \begin{align*}
        {\Pr}{\left[ \bigcup_{i \in N_b} B^0_i \neq B \mid n_b = s \right]} &= 
        {\Pr}{\left[ \exists b_j \in B \text{ such that } b_j \notin \bigcup_{i \in N_b} B^0_i \right]}\\
        & \leq \sum_{b_j \in B} {\Pr}{\left[ b_j \notin \bigcup_{i \in N_b} B^0_i \right]}\\
        & = \sum_{b_j \in B} \prod_{i \in N_b} {\Pr}{\left[ b_j \notin B^0_i \right]}\\
        & = \sum_{b_j \in B} \prod_{i \in N_b} \Big(1-{\Pr}{\left[ b_j \in B^0_i \right]}\Big)\\
        & = (m - \ell) \cdot \left(1-\frac{k}{m}\right)^{s}.    
    \end{align*}
    By plugging the values we have so far into the first equation, we get:
    \begin{align*}
        {\Pr}{\left[ \bigcup_{i \in N_b} B^0_i \neq B \right]} & = 
        \sum^n_{s=0} {\Pr}{\left[ \bigcup_{i \in N_b} B^0_i \neq B \mid n_b = s \right]} \cdot {\Pr}{\left[ n_b = s \right]}\\
        & \leq \sum^n_{s=0} (m - \ell) \cdot \left(1-\frac{k}{m}\right)^{s} \cdot 
        {n \choose s} (1-p)^x p^{n-x}\\
        & = (m - \ell) \cdot \sum^n_{s=0} {n \choose s} \left((1-p) \left(1-\frac{k}{m}\right)\right)^x p^{n-s}\\
        & = (m - \ell) \cdot \left( p + (1-p) \left(1-\frac{k}{m}\right) \right)^n\\
        & = (m - \ell) \cdot \left( 1 - (1-p) \frac{k}{m} \right)^n.
    \end{align*}
    Since $1 - (1-p)\nicefrac{k}{m} \in (0, 1)$, it follows that 
    $\left( 1 - (1-p) \nicefrac{k}{m} \right)^n$ goes to $0$ as $n \rightarrow \infty$.
    Therefore:
    $$
        \lim_{n \rightarrow \infty} {\Pr}{\left[ \bigcup_{i \in N_b} B^0_i \neq B \right]} = 0,
    $$
    which is what we wanted to show.
\end{proof}

\begin{proof}[Lemma \ref{lemma:round-1-agent-accuracies}]
    ($\textsf{a}$)
    The argument here is straightforward: 
    if $k \leq m-\ell$, the influx of newly disclosed $m-\ell$ evidence items for $b$
    matches or exceeds any amount of evidence for $a$ among the $k$ 
    items agents have from round $t=0$:
    no initial $a$-supporter survives ($p'_a = 0$), and at the end of $t=1$  
    each voter is either indifferent ($p'_\emptyset > 0$) or votes for $b$.
    If $k < m-\ell$ there is no possibility of ties, 
    so $p'_\emptyset = 0$ and $p'_b = 1$.

    ($\textsf{b}$)
    We start by showing that $p'_b \geq \nicefrac{1}{2}$.
    Agent $i$ votes for $b$ just in case $i$'s has more evidence for
    $b$ than for $a$ at the end of round $1$.
    Since at the end of this round $i$ has all $m-\ell$ evidence items for $b$
    and all the initial evidence items for $a$ (no evidence for $a$ is disclosed at $t=1$),
    this is equivalent to:
    $$
        p'_b = {\Pr}\Big[ |A_i^0| < m-\ell \Big] = {\Pr}\Big[ |A^0_i|\leq m-\ell-1 \Big].
    $$   
    Note that the second inequality of 
    Lemma \ref{lemma:hypergeometric-median-inequality} ($M \leq \lceil \mu \rceil$), 
    together with the definition of median $M$, implies that:
    $$
        \frac{1}{2} \leq 
        {\Pr}\Big[ |A_i^0| \leq M \Big] \leq 
        {\Pr}\Big[ |A_i^0| \leq \lceil \mu \rceil \Big],
    $$
    where $\mu = k \cdot\nicefrac{\ell}{m}$ is the expected value of the 
    hypergeometric variable $|A_i^0|$ and $M$ is its median.
    To ensure the probability of a vote for $b$ at $t=1$ is above $\nicefrac{1}{2}$,
    we push $m-\ell-1$ above $\lceil \mu \rceil$. More precisely, 
    note that $k < \nicefrac{m(m-\ell-1)}{\ell}$ is equivalent to:
    $$
        \mu = k\frac{\ell}{m} < m-\ell-1,
    $$
    which, in turn, implies that $\lceil \mu \rceil \leq m-\ell-1$.
    Chaining this with the inequality $M \leq \lceil \mu \rceil$ from 
    Lemma \ref{lemma:hypergeometric-median-inequality} we get that $M \leq m-\ell-1$, and thus:
    $$
        p'_b = 
        {\Pr}\Big[ |A_i^0| \leq m-\ell-1 \Big] \geq 
        {\Pr}\Big[ |A_i^0| \leq M \Big] 
        \geq \frac{1}{2}.
    $$
    The probability that $i$ votes for $a$ at $t=1$ is:
    $$
        p'_a = {\Pr}\Big[ |A_i^0| > m-\ell \Big],
    $$
    i.e., the event where $i$ starts out with more than $m-\ell$ evidence for $a$,
    which ensures that $i$ keeps their opinion for $a$ even after the influx of $b$-evidence at $t=1$.
    Since $k < m-\ell$, the probability of this event is non-zero.
    There only other remaining possibility is that $i$ finishes round $1$ indifferent:
    $$
        p'_\emptyset = {\Pr}\Big[ |A_i^0| = m-\ell \Big],
    $$
    which is also non-zero.
    We have already shown that the probability of a vote for $b$ at $t=1$ is above $\nicefrac{1}{2}$,
    so it follows that ${\Pr}\big[ v_i^1 = a \big] < \nicefrac{1}{2}$
    and ${\Pr}\big[ v_i^1 = b \big] < 1$.

    ($\textsf{c}$)
    The proof here is analogous to that of ($ii$), 
    and is obtained using the first inequality of 
    Lemma \ref{lemma:hypergeometric-median-inequality} ($\lfloor \mu \rfloor \leq M$)
    and by pushing $m-\ell+1$ below $\lfloor \mu \rfloor$.
\end{proof}

\begin{proof}[Proposition \ref{proposition:round-1-collective-outcome}]
    ($\textsf{a}$)
    If $k \leq m-\ell$, then, 
    by Lemma \ref{lemma:round-1-agent-accuracies}($\textsf{a}$),
    every agent is either indifferent 
    or votes for $b$ at the end of round $1$.
    The probability that everyone is indifferent goes to $0$ as $n \rightarrow \infty$,
    so we can be certain that there is at least one $b$-voter to break the deadlock.
    Thus, as $n \rightarrow \infty$, the collective vote at the end of round $1$ is $b$ with probability $1$.
    Since there are no $a$-voters remaining, deliberation terminates at this round with $b$ as winner.
    
    ($\textsf{b}$)
    If $m-\ell < k < \nicefrac{m(m-\ell-1)}{\ell}$, 
    then Lemma \ref{lemma:round-1-agent-accuracies}($\textsf{b}$) guarantees that, 
    at the end of round $1$,
    agents are more likely to vote for $b$ than for $a$, 
    and thus the collective vote at the end of round $1$ is $b$ with high probability as $n$ grows.
    Note that this marks a reversal in the majority opinion from the initial round, where $a$ had the majority.
    
    ($\textsf{c}$)
    If $\nicefrac{m(m-\ell+1)}{\ell} < k < 2(m-\ell)$, 
    then, by Lemma \ref{lemma:round-1-agent-accuracies}($\textsf{c}$),
    agents are more likely to vote for $a$ 
    than for $b$ at the end of round $1$, and thus the collective vote at the end of round $1$ 
    is $a$ with high probability as $n$ grows. 
    Since all the available evidence for $b$ is already public, 
    there is no new information in favor of $b$ to be disclosed in the next round, 
    and thus deliberation terminates at this round with $a$ as winner.
\end{proof}

\begin{proof}[Lemma \ref{lemma:round-2-agent-accuracies}]
    By Proposition \ref{proposition:round-1-collective-outcome}($\textsf{b}$),
    the collective vote at the end of round $1$ is $b$.
    By Lemma \ref{lemma:round-1-agent-accuracies}($\textsf{b}$),
    there is also a non-zero fraction of $a$-voters, 
    who are now in the minority and thus eager to disclose their evidence for $a$.
    By the same argument as in Lemma \ref{lemma:initial-minority-covers-b-evidence},
    as $n \rightarrow \infty$ the fraction of $a$-voters at the end of round $1$ 
    collectively hold all the available evidence for $a$.
    When this evidence gets revealed at $t=2$ everyone ends
    up knowing all $\ell$ evidence items for $a$.
    By assumption, there is more evidence for $a$ than for $b$ ($\ell > m-\ell$), 
    so the influx of $a$-evidence at $t=2$ ensures that every agent votes
    ends up voting for $a$, i.e., $p''_a = 1$.
\end{proof}

\begin{proof}[Proposition \ref{proposition:round-2-collective-outcome}]
    By Lemma \ref{lemma:round-2-agent-accuracies}, 
    as $n \rightarrow \infty$, every agent votes for $a$ at the end of round $2$,
    so we have a unanimous vote for $a$.
    At this point all evidence is public, and deliberation terminates with $a$ as winner.
\end{proof}
\end{document}